\documentclass{elsarticle}

\usepackage{hyperref}
\usepackage{xcolor}
\usepackage{amsmath,amsfonts,amssymb,amsthm,mathtools}
\usepackage{mathrsfs}
\usepackage{enumitem}
\usepackage{booktabs}
\usepackage{array}
\usepackage{multirow}
\usepackage{listings}
\usepackage{tikz}
\usetikzlibrary{arrows,calc,matrix,positioning}
\usepackage[algo2e,algoruled]{algorithm2e}

\journal{Journal of Computational and Applied Mathematics}

\newcommand{\qedfill}{\hfill\ensuremath{\triangleleft}}
\renewcommand{\Re}{\mathbb{R}}
\newcommand{\Fi}{\mathbb{F}}
\newcommand{\PhiMat}{\Phi_{m,p,n}}
\newcommand{\BigO}{\mathcal{O}}
\newcommand{\Bil}{\mathrm{Bil}}
\newcommand{\trace}{\mathrm{trace}}
\newcommand{\Perm}{\mathfrak{S}}
\newcommand{\Gl}{\mathrm{GL}}
\newcommand{\UU}{\mathrm{\mathbf{U}}}
\newcommand{\VV}{\mathrm{\mathbf{V}}}
\newcommand{\WW}{\mathrm{\mathbf{W}}}
\renewcommand{\AA}{\mathrm{\mathbf{A}}}
\newcommand{\BB}{\mathrm{\mathbf{B}}}
\newcommand{\CC}{\mathrm{\mathbf{C}}}
\newcommand{\PP}{\mathrm{\mathbf{P}}}
\newcommand{\QQ}{\mathrm{\mathbf{Q}}}
\newcommand{\RR}{\mathrm{\mathbf{R}}}
\newcommand{\II}{\mathrm{\mathbf{I}}}
\newcommand{\MM}{\mathrm{\mathbf{M}}}
\newcommand{\cl}{\mathrm{cl}_\oplus}
\newcommand{\colspan}{\mathrm{colspan}}
\newcommand{\rank}{\mathrm{rank}}
\renewcommand{\aa}{\mathrm{\mathbf{a}}}
\newcommand{\qq}{\mathrm{\mathbf{q}}}
\newcommand{\GGc}{\mbox{\upshape\sffamily G}}
\newcommand{\diag}{\mathrm{diag}}
\newcommand{\SSc}{\mathscr{S}}
\newcommand{\xx}{\mathrm{\mathbf{x}}}
\newcommand{\yy}{\mathrm{\mathbf{y}}}

\newcommand{\XX}{\mathrm{\mathbf{X}}}

\newcommand{\ZZ}{\mathrm{\mathbf{Z}}}
\newcommand{\dime}{\mathrm{dim}}
\newcommand{\AAt}{\tilde{\AA}}
\newcommand{\spann}{\mathrm{span}}
\newcommand{\UUt}{\tilde{\UU}}
\newcommand{\VVt}{\tilde{\VV}}
\newcommand{\WWt}{\tilde{\WW}}
\newcommand{\vecc}{\mathrm{vec}}
\newcommand{\Inj}{\mbox{\upshape\sffamily Inj}}
\newcommand{\Range}{\mathrm{Range}}
\newcommand{\ZZb}{\mathbb{Z}}
\newcommand{\ND}{\mbox{\upshape\sffamily ND}}

\newcommand{\PhiT}{\tilde{\Phi}}
\newcommand{\myF}{{[F]}}
\newcommand{\round}{\mathrm{round}}

\theoremstyle{plain}
\newtheorem{theorem}{Theorem}[section]
\newtheorem{proposition}[theorem]{Proposition}
\newtheorem{lemma}[theorem]{Lemma}
\theoremstyle{definition}
\newtheorem{definition}{Definition}[section]
\newtheorem{assumption}{Assumption}[section]
\theoremstyle{remark}
\newtheorem{remark}{Remark}[section]
\newtheorem{example}{Example}[section]

\begin{document}

%%%%%%%%%%%%%%%%%%%%%%%%%%%%%%%%%%%%%%%%%%%%%%%%%%%%%%%%%%%%%%%%%%%%%%%%%%%%%%%%%%%%%%%%%%%%%%%%%%%%
\begin{frontmatter}

\title{Equivalent Polyadic Decompositions of\\Matrix Multiplication Tensors\tnoteref{mytitlenote}}
\tnotetext[mytitlenote]{The work of the G.~Berger was supported by (i) the Fonds de la Recherche Scientifique -- FNRS and the Fonds Wetenschappelijk Onderzoek -- Vlaanderen under EOS Project no 30468160, (ii) ``Communaut\'e fran\c{c}aise de Belgique -- Actions de Recherche Concert\'ees'' (contract ARC 14/19-060).
L.~De Lathauwer and M.~Van Barel's research is funded by (i) the Research Council KU Leuven, C1-project c16/15/059-nD (Numerical Linear Algebra and Polynomial Computations), and by (ii) the Fund for Scientific Research--Flanders (Belgium), EOS Project no 30468160 (SeLMA).
R.~Jungers is supported by the Walloon Region and the Innoviris Foundation.}

%% Group authors per affiliation:
% \author{Elsevier\fnref{myfootnote}}
% \address{Radarweg 29, Amsterdam}
% \fntext[myfootnote]{Since 1880.}

%% or include affiliations in footnotes:
\author[address1]{Guillaume~O.~Berger\corref{mycorrespondingauthor}\fnref{myfootnote}}
\cortext[mycorrespondingauthor]{Corresponding author}
\ead{guillaume.berger@uclouvain.be}

\author[address1]{P.-A.~Absil}
\ead{pa.absil@uclouvain.be}

\author[address2,address3]{Lieven~De~Lathauwer}
\ead{lieven.delathauwer@kuleuven.be}

\author[address1]{Rapha\"el~M.~Jungers\fnref{myfootnote}}
\ead{raphael.jungers@uclouvain.be}

\author[address4]{Marc~Van~Barel}
\ead{marc.vanbarel@cs.kuleuven.be}

\address[address1]{ICTEAM Institute, UCLouvain, 1348 Louvain-la-Neuve, Belgium}
\address[address2]{Group of Science, Engineering and Technology, KU Leuven Kulak, 8500 Kortrijk, Belgium}
\address[address3]{Department of Electrical Engineering (ESAT), KU Leuven, 3001 Leuven, Belgium}
\address[address4]{Department of Computer Science, KU Leuven, 3001 Leuven, Belgium}

\fntext[myfootnote]{G.~Berger is an FNRS/FRIA Fellow.
R.~Jungers is an FNRS Research Associate}

\begin{abstract}
Invariance transformations of polyadic decompositions of matrix multiplication tensors define an equivalence relation on the set of such decompositions.
In this paper, we present an algorithm to efficiently decide whether two polyadic decompositions of a given matrix multiplication tensor are equivalent.
With this algorithm, we analyze the equivalence classes of decompositions of several matrix multiplication tensors.
This analysis is relevant for the study of fast matrix multiplication as it relates to the question of how many essentially different fast matrix multiplication algorithms there exist.
This question has been first studied by de~Groote, who showed that for the multiplication of $2\times2$ matrices with $7$ active multiplications, all algorithms are essentially equivalent to Strassen's algorithm.
In contrast, the results of our analysis show that for the multiplication of larger matrices, (e.g., $2\times3$ by $3\times2$ or $3\times3$ by $3\times3$ matrices), two decompositions are very likely to be essentially different.
We further provide a necessary criterion for a polyadic decomposition to be equivalent to a polyadic decomposition with integer entries.
Decompositions with specific integer entries, e.g., powers of two, provide fast matrix multiplication algorithms with better efficiency and stability properties.
This condition can be tested algorithmically and we present the conclusions obtained for the decompositions of small/medium matrix multiplication tensors.
\end{abstract}


\end{frontmatter}
%%%%%%%%%%%%%%%%%%%%%%%%%%%%%%%%%%%%%%%%%%%%%%%%%%%%%%%%%%%%%%%%%%%%%%%%%%%%%%%%%%%%%%%%%%%%%%%%%%%%

%%%%%%%%%%%%%%%%%%%%%%%%%%%%%%%%%%%%%%%%%%%%%%%%%%%%%%%%%%%%%%%%%%%%%%%%%%%%%%%%%%%%%%%%%%%%%%%%%%%%
\section{Introduction}\label{sec-intro}

The straightforward way to multiply two $N\times N$ matrices costs $\BigO(N^3)$ operations.
In particular, multiplying two $2\times 2$ matrices requires $8$ scalar multiplications.
However, as first remarked by V.~Strassen \cite{Str69} in 1969, the arithmetic operations can be grouped cleverly to reduce the work to $7$ multiplications only.
By doing this recursively, we can reduce the cost for the multiplication of $N\times N$ matrices to $\BigO(N^{2.81})$ operations.
Strassen's discovery opened the door to a considerable amount of research on the algorithmic complexity of matrix multiplication (see paragraphs below for a survey).
The reduction of the complexity may actually become so significant that a new architecture for large matrix multiplication is emerging.
Essential is first that we find inexpensive schemes for the multiplication of relatively small matrices.

The multiplication of $m\times p$ matrices by $p\times n$ matrices can be represented by a third-order tensor.
Finding inexpensive schemes for the multiplication of such matrices can be approached by decomposing the associated tensor as a sum of rank-$1$ terms (such a decomposition is called \emph{polyadic decomposition}).
The minimal number of rank-$1$ terms necessary to decompose a tensor is its rank.
In the case of matrix multiplication, the rank of the associated tensor is equal to the smallest number of active multiplications needed to compute the matrix product.
(By \emph{active multiplication}, we mean a multiplication of two scalars that both depend on the matrices to be multiplied.)
As a consequence, determining the rank of the associated tensor allows us to find an exponent $\alpha$ such that the complexity for the multiplication of $N\times N$ matrices is at most of $\BigO(N^\alpha)$ arithmetic operations \cite{BurCla13}.

Although the problem of matrix multiplication complexity is quite old, only partial results are known so far.
Even for the multiplication of small matrices, determining the rank of the associated tensor is still an open problem.
The largest case that is completely understood is the multiplication of $2\times 2$ matrices by $2\times 2$ matrices.
The rank of the associated tensor is $7$ \cite{BroDob78} (so that Strassen's algorithm is optimal), and it was proved by de Groote \cite{deG78b} that the decomposition induced by Strassen's algorithm is essentially unique (with respect to a class of transformations acting on polyadic decompositions of matrix multiplication; see the paragraph hereunder).
For the multiplication of $3\times3$ matrices, an algorithm using $23$ active multiplications was proposed by Laderman \cite{Lad76} in 1976, and by Makarov in 1986 \cite{Mak86}; see also \cite{Sed17b}.
This means that the rank of the associated tensor is at most $23$.
On the other hand, Bl{\"a}ser proved \cite{Bla03} in 2003 that the rank for the multiplication of $3\times3$ matrices should be at least $19$.
The gap $19$--$23$ has not been reduced since then.
Other algorithms for the multiplication of matrices with small and medium sizes are proposed in \cite{Mak87}, \cite{Sed17a}, \cite{Sed17c}, \cite{Ros19}; also the paper \cite{BalBen16} of Ballard et al.~gives a good overview of the practical algorithms that are available in 2016.
Further reductions of the exponent of matrix multiplication complexity have been achieved, e.g., by Pan \cite{Pan78,Pan80}, Bini et al.~\cite{BinCap79}, Sch\"onhage \cite{Sch81}, and Coppersmith and Winograd \cite{CopWin82} by means of more advanced techniques (including namely the study of the ``border rank'' of the associated tensor; see also \cite{Sto10,LanMic17}).
Currently, the best known upper bound for the complexity of matrix multiplication is $\BigO(N^{2.3729})$ by Le Gall \cite{LeG14}.

The aim of this paper is to study the connections between polyadic decompositions for matrix multiplication tensors.
In particular, we consider three types of transformations, called invariance transformations, acting on the set of polyadic decompositions of a given matrix multiplication tensor, and we study the equivalence relation induced by these transformations.
These transformations have been studied by de Groote \cite{deG78a,deG78b} who has shown that Strassen's algorithm is essentially unique in the sense that every other decomposition with $7$ rank-$1$ terms is equivalent to it.
In contrast, for the multiplication of $3\times3$ matrices, Johnson and McLoughlin \cite{JohMcL86} showed that Laderman's algorithm \cite{Lad76} is not essentially unique.
They provided two parametrized families of decompositions (with $23$ rank-$1$ terms) of the $3\times3$ matrix multiplication tensor that are mutually inequivalent and also inequivalent to Laderman's.
Later, Oh, Kim and Moon \cite{OhJin13} discovered other decompositions of the $3\times3$ matrix multiplication tensor inequivalent to the previous ones, and Sedoglavic \cite{Sed17b} showed that Laderman's algorithm can be constructed using Strassen's algorithm and related tensor's transformations.

The techniques used by de~Groote, Johnson and McLoughlin, and Oh et al.~to prove the equivalence/inequivalence of decompositions are either too specific \cite{deG78b,JohMcL86} or too conservative \cite{OhJin13} (some inequivalent decompositions are not recognized as such) to be applied to general decompositions of matrix multiplication tensors of arbitrary size.
In this paper, we present an algorithm for deciding whether two given decompositions are equivalent through invariance transformations.
Thanks to this algorithm, we were able to study the equivalence classes of large sample sets of matrix multiplication tensor decompositions (computed with numerical methods, see Section~\ref{sec-num}).
This allows us to get a better understanding of the equivalence relation of decompositions: for instance, the numerical experiments (Section~\ref{sec-num-equi}) suggest that for tensors larger than the $2\times2$ by $2\times2$ case, two ``generic'' decompositions are inequivalent.

In addition, we describe a necessary criterion for a matrix multiplication tensor decomposition to be \emph{discretizable}, that is, to be equivalent to a decompositions whose rank-$1$ terms can be factorized into vectors or matrices whose entries only take a few distinct values (for instance, we may want that all entries of the factor vectors/matrices of the rank-$1$ terms belong to the set $\{-1,0,+1\}$).
Such decompositions are called discrete decompositions \cite{Smi13}.
Our interest in discretizable decompositions originates from the observation that for small/medium matrix multiplication, the decompositions proposed in the literature are generally discrete: Strassen's and Laderman's algorithms are discrete with factor matrices coefficients belonging to $\{-1,0,+1\}$, other (inequivalent) decompositions with coefficients in $\{-1,0,+1\}$ are proposed, e.g., in \cite{BalBen16,OhJin13,TicPha17} for the multiplication of $2\times3$ by $3\times2$, $2\times3$ by $3\times3$, and $3\times3$ by $3\times3$ matrices.
In particular, all the decompositions listed in \cite{BalBen16} are discrete.

Discrete decompositions provide matrix multiplication algorithms with better efficiency and stability properties.
However, the classical iterative processes for computing tensor decompositions do not lead in general to solutions of this kind.
A reasonable approach to compute discrete decompositions is then to (i)~compute a general decomposition, and (ii)~use invariance transformations to obtain an equivalent discrete decomposition.
Closely related methods are used, e.g., in \cite{JohMcL86,TicPha17}.
The necessary criterion for discretizability allows us to identify some decompositions that cannot be transformed via invariance transformations into a discrete decomposition with a given ``target set'' for the coefficients.
By applying the necessary criterion to the sample sets of decompositions, we observed that, contrary to what the decompositions available in the literature suggest, most of the decompositions for tensors larger than the $2\times2$ by $2\times2$ case are not discretizable with respect to the commonly-used target sets (e.g., $\{0,\pm1\}$ or $\{0,\pm1/2,\pm1\}$).

The paper is organized as follows.
In Section~\ref{sec-preli}, we introduce the notation and recall the definitions of matrix multiplication tensors and polyadic decompositions.
Invariance transformations and the induced equivalence relations are introduced in Section~\ref{sec-invar}.
In Section~\ref{sec-equiv}, we describe the algorithm for deciding whether two decompositions are equivalent and if so, computing the invariance transformations involved in their equivalence.
The necessary criterion for discretizability is discussed in Section~\ref{sec-discrete}.
Numerical experiments are presented in Section~\ref{sec-num}.

%%%%%%%%%%%%%%%%%%%%%%%%%%%%%%%%%%%%%%%%%%%%%%%%%%%%%%%%%%%%%%%%%%%%%%%%%%%%%%%%%%%%%%%%%%%%%%%%%%%%
\section{Preliminaries}\label{sec-preli}

%%%%%%%%%%%%%%%%%%%%%%%%%%%%%%%%%%%%%%%%%%%%%%%%%%%%%%%%%%%%%%%%%%%%%%%%%%%%%%%%%%%%%%%%%%%%%%%%%%%%
\subsection{Matrix multiplication tensors and polyadic decompositions}

Let $U$, $V$ and $W$ be vector spaces over a field $\Fi$.
We denote by $\Bil(U,V;W)$ the set of $\Fi$-bilinear maps from $U\times V$ to $W$.
For positive integers $m,p,n$, the multiplication of $m\times p$ matrices by $p\times n$ matrices can be represented by the bilinear map $\PhiMat\in\Bil(\Fi^{m\times p},\Fi^{p\times n};\Fi^{m\times n})$ defined by
\[
\PhiMat(\AA,\BB)=\AA\BB.
\]
From the identification between multilinear maps and tensors, $\PhiMat$ is sometimes referred to as the \emph{$(m,p,n)$ matrix multiplication tensor}.

The concept of rank of a bilinear map $\Phi\in\Bil(U,V;W)$ is central in the analysis of the asymptotic complexity of matrix multiplication.
We say that $\Phi\neq0$ has rank $1$ if $\Phi(u,v) = f(u)g(v)w$ for some $f\in U^*$, $g\in V^*$ and $w\in W$, where $U^*$ and $V^*$ are the dual spaces of $U$ and $V$ respectively.
For a general $\Phi\in\Bil(U,V;W)$, an \emph{$F$-term polyadic decomposition} (in short $F$-PD) of $\Phi$ is a decomposition of $\Phi$ as the sum of $F$ rank-$1$ terms \cite{Hit27,KolBad09}:
\begin{equation}\label{eq-FPD}
\Phi(u,v) = \sum_{r=1}^F f_r(u) g_r(v) w_r
\end{equation}
for some $f_r\in U^*$, $g_r\in V^*$ and $w_r\in W$.
The \emph{rank} of $\Phi$ is the smallest $F$ such that $\Phi$ admits an $F$-term polyadic decomposition \eqref{eq-FPD}.

For a matrix multiplication tensor $\PhiMat$, a polyadic decomposition like \eqref{eq-FPD} requires $f_r\in(\Fi^{m\times p})^*$, $g_r\in(\Fi^{p\times n})^*$ and $w_r=\WW_r\in\Fi^{m\times n}$.
We may identify $f_r$ with the unique matrix $\UU_r\in\Fi^{p\times m}$ such that
\begin{equation}\label{eq-trace}
f_r(\AA) = \trace(\UU_r\AA) = \sum_{i=1}^p\sum_{j=1}^m \UU_r^{(i,j)} \AA^{(j,i)}
\end{equation}
for every $\AA\in\Re^{m\times p}$, where $\MM^{(i,j)}$ denotes the $(i,j)$th entry of a matrix $\MM$.
In the same way, $g_r$ can be identified with a unique matrix $\VV_r\in\Fi^{n\times p}$.
If $\UU_1,\ldots,\UU_F$, $\VV_1,\ldots,\VV_F$ and $\WW_1,\ldots,\WW_F$ give rise to a decomposition \eqref{eq-FPD} of $\PhiMat$, we will say with slight abuse of notation that the triple $(\UU_\myF,\VV_\myF,\WW_\myF)$\footnote{\label{foo-F-uple}The rational behind this notation is that if $\XX_1,\XX_2,\ldots,\XX_F$ are $F$ mathematical objects, then $\XX_\myF$ denotes the ordered set, or $F$-uple, $(\XX_1,\XX_2,\ldots,\XX_F)$.} is an $F$-term polyadic decomposition ($F$-PD) of $\PhiMat$.

The link between matrix multiplication complexity and the rank of matrix multiplication tensors is nicely explained in \cite[Chapter~15]{BurCla13}.
Especially, it is shown how to build from an $F$-term polyadic decomposition ($F$-PD) of $\PhiMat$ a recursive algorithm for the multiplication of $N\times N$ matrices over $\Fi$ with complexity in $\BigO(N^{\omega+\varepsilon})$ arithmetic operations $\{+,-,\times\}$, with $\omega=3\log_{mpn}(F)$ and for any $\varepsilon>0$.
For instance, Strassen's algorithm can be obtained from a decomposition of $\Phi_{2,2,2}$ with $7$ terms.
This directly gives the well-known upper bound $\omega=3\log_8(7)\approx 2.81$ for the exponent of matrix multiplication complexity \cite{Str69}.

%%%%%%%%%%%%%%%%%%%%%%%%%%%%%%%%%%%%%%%%%%%%%%%%%%%%%%%%%%%%%%%%%%%%%%%%%%%%%%%%%%%%%%%%%%%%%%%%%%%%
\subsection{Discrete decompositions}

In this paper, we focus on algorithms for matrix multiplication over the field of real numbers, i.e., on the case $\Fi=\Re$ and $\UU_r$, $\VV_r$ and $\WW_r$ are real matrices.

For the problem of matrix multiplication over $\Re$ (or $\mathbb{C}$), two decompositions might be not equally useful even if they have the same number of rank-$1$ terms.
For instance, decompositions with ``structured'' values in the rank-$1$ terms are more useful in practice.
This leads us to the following definition:

\begin{definition}\label{def-discr}
A decomposition $(\UU_\myF,\VV_\myF,\WW_\myF)$ of $\PhiMat$ is said to be \emph{discrete} if the entries of $\UU_r$, $\VV_r$ and $\WW_r$ belong to $q\ZZb$ for some $q\in\Re$.
\end{definition}

Discrete decompositions are favorable for two reasons.
The first reason concerns the exactness of the decomposition: if $(\UU_\myF,\VV_\myF,\WW_\myF)$ is computed with numerical methods, then this will give a decomposition of $\PhiMat$ only up to some \emph{finite} accuracy (and also up to machine precision, due to floating-point arithmetic computations).
Hence, the matrix multiplication algorithm obtained from this decomposition will compute the product of $\AA\in\Re^{m\times p}$ and $\BB\in\Re^{p\times n}$ with a small error, even in exact arithmetic.
This is not advisable because this error will accumulate when we will apply the algorithm in a recursive way to compute the product of general $N\times N$ matrices \cite[Chapter~15]{BurCla13}.
These limited-accuracy issues can be overcome if we know a priori that $\UU_r$, $\VV_r$ and $\WW_r$ have their entries in a known discrete set.

The second reason to favor discrete decompositions is that the obtained algorithm for matrix multiplication will have better stability and computational cost properties.
Indeed, if the entries of $\UU_r$, $\VV_r$ and $\WW_r$ belong to $q\ZZb$, then it is not hard to show (we do not go into the details) that, modulo some pre- and post multiplication of $\AA\in\Re^{m\times p}$ and $\BB\in\Re^{p\times n}$ by $q$, the product $\AA\BB$ can be computed using only additions and multiplications of the entries of $\AA$ and $\BB$ by integers.
(For example, in Strassen's algorithm, $f_r(\AA)$ [resp.\ $g_r(B)$] can be obtained using only additions and subtractions of the entries of $\AA$ [resp.\ $\BB$].)
Multiplication by integers is more rapid and stable than multiplication by arbitrary floating-point numbers (for instance, multiplication by a power of $2$ is equivalent to changing the exponent in the floating point representation).
For more detailed information on the forward normwise error induced by a fast matrix multiplication algorithm, we refer the interested reader to \cite{BalBen16}.

%%%%%%%%%%%%%%%%%%%%%%%%%%%%%%%%%%%%%%%%%%%%%%%%%%%%%%%%%%%%%%%%%%%%%%%%%%%%%%%%%%%%%%%%%%%%%%%%%%%%
\section{Invariance transformations}\label{sec-invar}

The main goal of this paper is to study relations between decompositions of a given matrix multiplication tensor.
We describe three types of operations that transform an $F$-PD of a matrix multiplication tensor into another $F$-PD of the same tensor.
These transformations will be referred to as \emph{invariance transformations}.

\begin{proposition}[Invariance transformations]\label{pro-invariance}
Let $(\UU_\myF,\VV_\myF,\WW_\myF)$ be an $F$-PD of $\PhiMat$.
The following transformations produce matrices $\UU'_r$, $\VV'_r$ and $\WW'_r$ ($1\leq r\leq F$) such that $({\UU'}_\myF,{\VV'}_\myF,{\WW'}_\myF)$ is also an $F$-PD of $\PhiMat$:
\begin{itemize}[topsep=3pt,itemsep=2pt,leftmargin=*]
	\item Permutation transformations: let $\sigma\in\Perm_F$ (where $\Perm_F$ is the set of permutations of $\{1,\ldots,F\}$) and define
	\[
	\UU'_r = \UU_{\sigma(r)}^{} ,\quad \VV'_r = \VV_{\sigma(r)}^{} ,\quad \WW'_r = \WW_{\sigma(r)}^{}.
	\]
	\item Scaling transformations: choose coefficients $\lambda_r,\mu_r,\nu_r\in\Re$ such that $\lambda_r\mu_r\nu_r=1$ for each $1\leq r\leq F$ and define
	\[
	\UU'_r=\lambda_r^{}\UU_r^{} ,\quad \VV'_r=\mu_r^{}\VV_r^{} ,\quad \WW'_r=\nu_r^{}\WW_r^{} .
	\]
	\item Trace transformations: let $\PP\in\Gl(m)$, $\QQ\in\Gl(p)$ and $\RR\in\Gl(n)$ (where $\Gl(h)$ denotes the set of invertible $h\times h$ matrices), and define
	\[
	\UU'_r=\QQ^{-1}\UU_r^{}\PP ,\quad \VV'_r=\RR^{-1}\VV_r^{}\QQ ,\quad \WW'_r=\PP^{-1}\WW_r^{}\RR .
	\]
\end{itemize}\vskip0pt
\end{proposition}

The first two classes of invariance transformations (permutations and scaling) provide invariance transformations for the decompositions of any tensors.
However, the third class (trace transformations) is somehow specific to matrix multiplication tensors, as it originates from the invariance of the trace operator (see proof below); hence the name ``trace transformations''.

\begin{proof}[Proof of Proposition~\ref{pro-invariance}]
(See, e.g., \cite{deG78a}).
The invariance of the permutation and scaling transformations is straightforward.
For the trace transformations, let $f'_r\in(\Re^{m\times p})^*$, $g'_r\in(\Re^{p\times n})^*$ and $\Phi'\in\Bil(\Re^{m\times p},\Re^{p\times n};\Re^{m\times n})$ be given by \eqref{eq-trace} and \eqref{eq-FPD} with $({\UU'}_\myF,{\VV'}_\myF,{\WW'}_\myF)$.
Then
\[
f'_r(\AA)=\trace(\UU'_r\AA)=\trace(\UU_r^{}[\PP\AA\QQ^{-1}])=f_r^{}(\PP\AA\QQ^{-1})
\]
where we have used the invariance property of the trace operator with respect to cyclic permutations.
Similarly, $g'_r(\BB)=g_r^{}(\QQ\BB\RR^{-1})$.
It follows that
\begin{align*}
\Phi'(\AA,\BB) &= \PP^{-1} \bigg[ \sum_{r=1}^F f_r(\PP\AA\QQ^{-1}) g_r(\QQ\BB\RR^{-1}) \WW_r \bigg] \RR \\
&= \PP^{-1}\PhiMat(\PP\AA\QQ^{-1},\QQ\BB\RR^{-1}) \RR = \PP^{-1} (\PP\AA\QQ^{-1}) (\QQ\BB\RR^{-1}) \RR = \AA\BB.
\end{align*}
Thus $\Phi'=\PhiMat$ showing that $({\UU'}_\myF,{\VV'}_\myF,{\WW'}_\myF)$ is an $F$-PD of $\PhiMat$.
\end{proof}

Invariance transformations define an equivalence relation on the set of $F$-PDs of a given matrix multiplication tensor.
For fixed $m,p,n$ and $F$, two polyadic decompositions $(\UU_\myF,\VV_\myF,\WW_\myF)$ and $({\UU'}_\myF,{\VV'}_\myF,{\WW'}_\myF)$ of $\PhiMat$ are \emph{equivalent} if there exist permutation, scaling and/or trace transformations that allow one to transform $(\UU_\myF,\VV_\myF,\WW_\myF)$ into $({\UU'}_\myF,{\VV'}_\myF,{\WW'}_\myF)$.
We will also say that $(\UU_\myF,\VV_\myF,\WW_\myF)$ and $({\UU'}_\myF,{\VV'}_\myF,{\WW'}_\myF)$ are \emph{permutation-equivalent} if there exists a permutation transformation allowing us to transform $(\UU_\myF,\VV_\myF,\WW_\myF)$ into $({\UU'}_\myF,{\VV'}_\myF,{\WW'}_\myF)$.
Similarly, we define the notion of \emph{(scaling+trace)-equivalence}.

We have just seen that invariance transformations can be used to produce many different decompositions (i.e., many fast matrix multiplication algorithms) from a given one.
This raises the following questions that we will address in this paper:

\vspace{1ex plus .25ex minus .1ex}

{\itshape How many inequivalent polyadic decompositions does $\PhiMat$ admit?
In other words, how many essentially different fast matrix multiplication algorithms are there for the multiplication of $m\times p$ matrices by $p\times n$ matrices?}
We will tackle this question in Sections~\ref{sec-equiv} and~\ref{sec-num-equi}.

\vspace{1ex plus .25ex minus .1ex}

{\itshape Starting from a given algorithm for the multiplication of $m\times p$ matrices by $p\times n$ matrices, can we obtain with invariance transformations another algorithm with better performance (e.g., in terms of stability and efficiency)?}
We will tackle this question in Sections~\ref{sec-discrete} and~\ref{sec-num-disc}.

\begin{remark}
At first sight, it might look like the scaling transformations act as a particular case of the trace transformations with
\[
\PP = \left(\frac{\lambda}{\nu}\right)^{1/3} \II_m ,\quad \QQ = \left(\frac{\mu}{\lambda}\right)^{1/3} \II_p ,\quad \RR = \left(\frac{\nu}{\mu}\right)^{1/3} \II_n
\]
for example.
In fact, this is not the case since the above $\PP,\QQ,\RR$ will rescale all the matrices $\UU_r,\VV_r,\WW_r$ with the same coefficients $\lambda,\mu,\nu$ (provided $\lambda\mu\nu=1$) while the scaling transformations admit different coefficients $\lambda_1,\ldots,\lambda_F$, $\mu_1,\ldots,\mu_F$ and $\nu_1,\ldots,\nu_F$.~\qedfill
\end{remark}

%%%%%%%%%%%%%%%%%%%%%%%%%%%%%%%%%%%%%%%%%%%%%%%%%%%%%%%%%%%%%%%%%%%%%%%%%%%%%%%%%%%%%%%%%%%%%%%%%%%%
\section{An algorithm for checking equivalence}\label{sec-equiv}

In this section, we present an algorithm for deciding whether two $F$-PDs of a matrix multiplication tensor are equivalent.
Under mild assumptions on the input $F$-PDs, the algorithm will either return the permutation, scaling and trace transformations that allow one to connect both $F$-PDs or conclude that the two $F$-PDs are not equivalent to each other.
The working assumptions were satisfied for $100\%$ of the samples on which we performed numerical experiments (see Section~\ref{sec-num-equi}), motivating the qualifier ``mild'' assumptions.

We start this section by introducing the concept of clustering number of a matrix.
This number can be computed efficiently and is used in the assumption to guarantee proper working of the algorithm.

%%%%%%%%%%%%%%%%%%%%%%%%%%%%%%%%%%%%%%%%%%%%%%%%%%%%%%%%%%%%%%%%%%%%%%%%%%%%%%%%%%%%%%%%%%%%%%%%%%%%
\subsection{The clustering number of a matrix}

Let $\AA$ be an $m\times n$ matrix.
Let $\{U_1,\ldots,U_S\}$ be a family of linearly independent subspaces of $\Re^m$ (i.e., $u_1+\cdots+u_S=0$ with $u_i\in U_i$ for each $1\leq i\leq S$ implies $u_i=0$ for every $1\leq i\leq S$).
If each column of $\AA$ belongs to some $U_i$ (in fact, except the case where the column contains only zeros, it may belong to at most one $U_i$ since they are linearly independent), then we say that $\{U_1,\ldots,U_S\}$ is a \emph{cover} of $\AA$.
The largest integer $S^*$ such that there exists a cover of $\AA$ with $S^*$ linearly independent subspaces is called the \emph{clustering number} of $\AA$ and is denoted by $\cl(A) = S^*$ (the choice of the symbol $\oplus$ comes from the fact that if $\{U_1,\ldots,U_S\}$ is a cover of $A$ with $S=S^*$ then $U_1 \oplus \cdots \oplus U_S = \Re^m$).

\begin{example}
Let $U_1 = \colspan(\AA)$ and suppose that the rank of $\AA$ is $r$.
Then it is easy to build a cover $\{U_1,U'_1,\ldots,U'_{m-r}\}$ where the $U'_i$'s are one-dimensional subspaces.
Hence, we conclude that $\cl(\AA) \geq m+1-\rank(\AA)$.~\qedfill
\end{example}

Suppose that $\AA$ has full row-rank, i.e., $\rank(\AA) = m$.
We give a characterization of the clustering number of $\AA$ in terms of the connected components of a graph.
Denote by $\aa_1,\ldots,\aa_n$ the columns of $\AA$.
Without loss of generality, we may assume that the first $m$ columns of $\AA$ span $\Re^m$.
Let $\AA' = [\aa_1,\ldots,\aa_m]$ and $\AA'' = [\aa_{m+1},\ldots,\aa_n]$.

Define the undirected graph $\GGc$ as follows.
The integers $\{1,\ldots,m\}$ are the nodes of $\GGc$ and for each column $\aa_j$ of $\AA''$, let $\qq_j = (q_{j,1},\ldots,q_{j,m})^\top$ be its coordinates in the basis defined by $\AA'$, i.e., $\qq_j = (\AA')^{-1}\aa_j$.
For each $i_1,i_2 \in \{1,\ldots,m\}$, draw an edge between the nodes $i_1$ and $i_2$ if and only if there is a column $\aa_j$ of $\AA''$ such that its coordinate vector has a nonzero component in both $\aa_{i_1}$ and $\aa_{i_2}$, i.e., if $q_{j,i_1}\neq0$ and $q_{j,i_2}\neq0$ (clearly, there might be multiple edges and also loops).
Moreover, each edge receives a \emph{label}: this label is simply $j$, the index of the column of $\AA''$ that led to this edge.
An example is represented in Figure~\ref{fig-matr-grap}.
This construction allows us to state the following lemma:

\begin{figure}
\newcommand{\mycolsepmat}{4pt}%
\centering
\begin{tikzpicture}[%
auto,
node distance=3cm,
vertex/.style={circle,draw,thick}%
]
\matrix[%
inner sep=2pt,
nodes={inner sep=6pt},
matrix of math nodes,
nodes in empty cells,
left delimiter={[},
right delimiter={]}%
] (m) at (-7,-1)
{
1 & 0 & 0 & |[inner sep=\mycolsepmat]| & 1 & |[inner sep=\mycolsepmat]| & 0 & |[inner sep=\mycolsepmat]| & 1 \\
0 & 1 & 0 & |[inner sep=\mycolsepmat]| & 1 & |[inner sep=\mycolsepmat]| & 0 & |[inner sep=\mycolsepmat]| & 2 \\
0 & 0 & 1 & |[inner sep=\mycolsepmat]| & 0 & |[inner sep=\mycolsepmat]| & 1 & |[inner sep=\mycolsepmat]| & 0 \\
};
\draw (m-1-1.north west) rectangle (m-3-3.south east);
\draw (m-1-5.north west) rectangle (m-3-5.south east);
\draw (m-1-7.north west) rectangle (m-3-7.south east);
\draw (m-1-9.north west) rectangle (m-3-9.south east);
% -------------------------------------------------------------------
\node [vertex] (1) at (0,0) {1};
\node [vertex] (2) at (-1.62,-2.8) {2};
\node [vertex] (3) at (1.62,-2.8) {3};
\draw [thick] (1) to [out=-140,in=80] node [midway,above left] {4} (2);
\draw [thick] (1) to [out=-100,in=40] node [midway,below right] {6} (2);
\draw [thick] (1) to [out=-20,in=20,looseness=10] node [midway,right] {4} (1);
\draw [thick] (1) to [out=100,in=140,looseness=10] node [midway,above] {6} (1);
\draw [thick] (2) to [out=200,in=160,looseness=10] node [midway,left] {4} (2);
\draw [thick] (2) to [out=-80,in=-40,looseness=10] node [midway,below] {6} (2);
\draw [thick] (3) to [out=100,in=140,looseness=10] node [midway,above] {5} (3);
% -------------------------------------------------------------------
\draw [thick,-latex] (-4,-1) -- (-3,-1);
% -------------------------------------------------------------------
\node [below=5mm of m-3-1.south,inner sep=6pt] (j1) {$1$};
\node [below=5mm of m-3-2.south,inner sep=6pt] (j2) {$2$};
\node [below=5mm of m-3-3.south,inner sep=6pt] (j3) {$3$};
\node [below=5mm of m-3-5.south,inner sep=6pt] (j5) {$4$};
\node [below=5mm of m-3-7.south,inner sep=6pt] (j7) {$5$};
\node [below=5mm of m-3-9.south,inner sep=6pt] (j9) {$6$};
\draw [dashed] (j1.south west) rectangle (j3.north east) node [pos=0.5,anchor=center] (rect1) {};
\draw [dashed] (j5.south west) rectangle (j9.north east) node [pos=0.5,anchor=center] (rect2) {};
\node [below=7mm of rect1,anchor=base] {\sffamily indices $i$};
\node [below=7mm of rect2,anchor=base] {\sffamily indices $j$};
\end{tikzpicture}
\caption{Example of matrix $\AA$ and the associated graph $\GGc$.}%
\label{fig-matr-grap}
\end{figure}

\begin{proposition}\label{pro-clus-graph}
Let $\AA$ and $\GGc$ be defined as above.
Then the clustering number of $\AA$ is equal to the number of connected components of $\GGc$.
\end{proposition}

\begin{proof}
Let $\{\GGc_1,\ldots,\GGc_T\}$ be the connected components of $\GGc$.
First, we show that $\cl(\AA)\geq T$.
For each $1\leq t\leq T$, let $I_t$ be the set of all column indices involved in the nodes of $\GGc_t$.
(For example, considering the graph in Figure~\ref{fig-matr-grap}, we would have $I_1 = \{1,2\}$ and $I_2 = \{3\}$.)
For each $t$, define the subspace $U_t$ as the subspace spanned by the columns $\aa_i$ with $i\in I_t$.
By hypothesis on $\AA'$ having full rank, the subspaces $U_t$ satisfy $\Re^m=U_1\oplus\cdots\oplus U_T$.
We have to show that each column $\aa_j$ of $\AA''$ belongs to some $U_t$.

Therefore, we show that the label $j$ appears in the edges of at most one component $\GGc_t$.
Indeed, if $j$ appears in $\GGc_{t_1}$ and $\GGc_{t_2}$, then $\aa_j$ has a nonzero component in at least one node $i_1$ of $\GGc_{t_1}$ and one node $i_2$ of $\GGc_{t_2}$.
Hence, there must be an edge between $i_1$ and $i_2$ and thus $\GGc_{t_1}$ and $\GGc_{t_2}$ are connected, a contradiction.
Thus $\cl(\AA) \geq T$.

To show that $\cl(\AA)\leq T$, let $S = \cl(\AA)$ and let $\{U_1,\ldots,U_S\}$ be a cover of $\Re^m$.
For each $1\leq s\leq S$, let $I_s$ be the set of the indices of the columns of $\AA'$ belonging to $U_s$.
Since $\AA'$ has full rank, it is clear that $U_s = \spann(\{\aa_i\}_{i\in I_s})$.
We show that $\{I_1,\ldots,I_S\}$ defines connected components of $\GGc$.
Indeed, if there is an edge, say with label $j$, between nodes $i_1 \in I_{s_1}$ and $i_2 \in I_{s_2}$, then $\aa_j$ has nonzero components in $\aa_{i_1}$ and in $\aa_{i_2}$ and thus it belongs to the subspaces $U_{s_1}$ and $U_{s_2}$.
However, by definition, the subspaces $U_s$ are linearly independent.
Hence, we must have $s_1=s_2$.
We conclude that there are at least $S$ connected components in $\GGc$ and thus $T \geq \cl(\AA)$.
\end{proof}

Proposition~\ref{pro-clus-graph} above gives an efficient way to compute the clustering number of a matrix.
We describe below another way to efficiently compute the clustering number of a matrix using linear algebra only (and that will be useful in the proof of Lemma~\ref{lem-dim-equi} below).
To do this, let $\AA$ be an $m\times n$ matrix and consider the following linear system:
\begin{equation}\label{eq-lin-cluster}
\MM \AA = \AA \, \diag(\xi_1,\ldots,\xi_n)
\end{equation}
with variables $\MM\in\Re^{m\times m}$ and $\xi_{[n]} = (\xi_1,\ldots,\xi_n)\in\Re^n$.
Clearly the system \eqref{eq-lin-cluster} is homogeneous.
Let us denote by $\SSc$ the vector space of solutions $(\MM,\xi_{[n]})$ to \eqref{eq-lin-cluster}.

\begin{lemma}\label{lem-lin-cluster-full}
Let $\AA$ have full row-rank and no zero columns and let $\SSc$ be defined as above.
Then $\cl(\AA)=\dime(\SSc)$.
\end{lemma}

\begin{proof}
First note that if $\AA\in\Re^{m\times n}$ has full row-rank then $m\leq n$.
Let $\AA'$ and $\AA''$ be defined as previously.
Without loss of generality, we may again assume that $\AA'$ has full rank.
Let $\xi_{[n]}=(\xi_1,\ldots,\xi_n)$ be fixed.
Then we have
\[
\MM = \AA' \, \diag(\xi_1,\ldots,\xi_m)\,\AA'^{-1}
\]
whence $\MM$ is completely determined by $\xi_{[n]}$.
Reversely, if $\MM$ is known, then every $\xi_1,\ldots,\xi_n$ are also determined since $\AA$ has no zero columns.
Hence, we only have to compute the number of degrees of freedom in $\xi_{[n]}$ to compute the dimension of $\SSc$.

Let $\GGc$ be the graph associated to $\AA$.
Suppose that $i_1$ and $i_2$ are two nodes that are adjacent to each other with an edge labeled by $j$.
Then, we have
\[
\MM \aa_j = \big[\AA' \, \diag(\xi_1,\ldots,\xi_m)\,\AA'^{-1}\big] \aa_j = \sum_{i=1}^m \aa_i\xi_iq_{j,i} = \aa_j\xi_j = \xi_j \sum_{i=1}^m \aa_iq_{j,i} .
\]
Since $\aa_1,\ldots,\aa_m$ are linearly independent and $q_{j,i_1}\neq0$ and $q_{j,i_2}\neq0$, the only solution provides that $\xi_{i_1}=\xi_{i_2}=\xi_j$.
We conclude that if two nodes $i_1,i_2$ belong to the same connected component $\GGc_t$, then $\xi_{i_1}=\xi_{i_2}$.
Hence, using Lemma~\ref{pro-clus-graph}, the dimension of $\SSc$ is lower than or equal to $\cl(\AA)$.

On the other hand, let $S = \cl(\AA)$ and let $\{U_1,\ldots,U_S\}$ be a cover of $\AA$.
Then for each $1\leq s\leq S$, let $\MM_s$ be the projection on $U_s$, i.e.,
\[
\MM_s (\xx + \yy) = \xx \;\;\;\;\;\;\;\;\; \text{for all $\xx\in U_s$ and $\yy\in\bigoplus_{s' \neq s} U_{s'}$}.
\]
Let $(\eta_1,\ldots,\eta_S)$ be a fixed vector and define $\MM = \sum_{s=1}^S \eta_s \MM_s$.
For each $1\leq i\leq n$, let $\xi_i = \eta_{s_i}$ where $s_i$ is the unique index $1\leq s_i\leq S$ such that $\aa_i \in U_{s_i}$.
Then $(\MM,\xi_{[n]})$ is a solution of \eqref{eq-lin-cluster}.
Hence, the dimension of $\SSc$ is at least $\cl(\AA)$.
\end{proof}

We are now able to state and prove the main theorem for this subsection:

\begin{theorem}\label{thm-lin-cluster-gen}
Let $\AA$ be an $m\times n$ matrix with rank $r$ and let $Z$ be the number of zero columns in $\AA$.
Then the clustering number of $\AA$ and the dimension of the solution space $\SSc$ of \eqref{eq-lin-cluster} satisfy
\[
\dime(\SSc) = \cl(\AA) + (m-1)(m-r) + Z .
\]\vskip0pt
\end{theorem}

\begin{proof}
First we suppose that there are no zero columns in $\AA$.
Let $\XX \in \Gl(m)$ and observe that $\cl(\XX\AA)=\cl(\AA)$.
Consider the linear system
\begin{equation}\label{eq-lin-cluster-X}
\MM \XX\AA = \XX\AA \, \diag(\xi_1,\ldots,\xi_n)
\end{equation}
and note that $(\MM,\xi_{[n]})$ is a solution of \eqref{eq-lin-cluster-X} if and only if $(\XX^{-1}\MM\XX,\xi_{[n]})$ is a solution of \eqref{eq-lin-cluster}.
Hence, taking an appropriate matrix $\XX$, we may assume without loss of generality that the last $m-r$ rows of $\AA$ are zero.

Let $\AAt$ be the $r\times n$ matrix consisting of the first $r$ rows of $\AA$.
Then $\AAt$ has full row-rank and it is easy to check that $\cl(\AA) = \cl(\AAt) + m-r$.
The matrix $\MM$ may be partitioned into the following blocks:
\[
\newcommand{\mycolsepmat}{4pt}%
\MM =
\begin{tikzpicture}[baseline={([yshift=-3pt]m.center)}]
\matrix[%
inner sep=2pt,
nodes={inner sep=10pt},
matrix of math nodes,
nodes in empty cells,
left delimiter={[},
right delimiter={]}%
] (m) at (-7,-1.4)
{
& & |[inner sep=\mycolsepmat]| & \\
& & |[inner sep=\mycolsepmat]| & \\
|[inner sep=\mycolsepmat]| & |[inner sep=\mycolsepmat]| & |[inner sep=\mycolsepmat]| & |[inner sep=\mycolsepmat]| \\
& & |[inner sep=\mycolsepmat]| & \\
};
\draw (m-1-1.north west) rectangle node [midway] {$\MM_1$} (m-2-2.south east);
\draw (m-4-1.north west) rectangle node [midway] {$\MM_2$} (m-4-2.south east);
\draw (m-1-4.north west) rectangle node [midway] {$\MM_3$} (m-4-4.south east);
\end{tikzpicture}
\hspace{1cm} \text{with} \;\;\; \left\{
\begin{array}{r@{\:}c@{\:}l}
\MM_1 &\in& \Re^{r \times r} \\
\MM_2 &\in& \Re^{(m-r) \times r} \\
\MM_3 &\in& \Re^{m \times (m-r)} \\
\end{array}
\right. \;.
\]
It is clear that $(\MM,\xi_{[n]})$ is a solution of \eqref{eq-lin-cluster} if and only if
\begin{equation}\label{eq-lin-cluster-part}
\MM_1\AAt=\AAt\,\diag(\xi_1,\ldots,\xi_n) \;\;\;\;\text{and}\;\;\;\; \MM_2\AAt=0 .
\end{equation}
From Lemma~\ref{lem-lin-cluster-full} and the fact that $\AAt$ has full row-rank, the solutions $(\MM_1,\xi_{[n]},\MM_2)$ of \eqref{eq-lin-cluster-part} form a vector space with dimension $\cl(\AAt)$.
On the other hand, there are no constraints on $\MM_3\in\Re^{m\times (m-r)}$.
This proves the assertion when $Z=0$.

Finally, observe that appending a zero column to $\AA$ does not change its clustering number and also does not change the space of solutions $\MM$ of \eqref{eq-lin-cluster}.
The only thing that changes is that the coefficient $\xi_{n+1}$ affected to this zero column might take any value.
Hence, the dimension of $\SSc$ is increased by one.
This concludes the proof of the theorem.
\end{proof}

\begin{remark}
For the interested reader, let us mention that the clustering number has an interpretation in terms of matroids: considering $\AA$ as a \emph{linear matroid} with ground set given by the columns of $\AA$ \cite[Chapter~39]{Sch03}, then we can show that the clustering number of $\AA$ is in fact equal to the number of \emph{connected components} (in the matroid sense \cite[Chapter~39]{Sch03}) of the matroid $\AA$.
Dedicated softwares exist to compute the connected components of a matroid.
See, e.g., \cite{SageNotebook}.
In fact, in the case of a linear matroid, the implementation in \cite{SageNotebook} is equivalent to dynamically computing the connected components of the graph $\GGc$ in Proposition~\ref{pro-clus-graph}.~\qedfill
\end{remark}

%%%%%%%%%%%%%%%%%%%%%%%%%%%%%%%%%%%%%%%%%%%%%%%%%%%%%%%%%%%%%%%%%%%%%%%%%%%%%%%%%%%%%%%%%%%%%%%%%%%%
\subsection{Computation of the scaling and trace transformations}

Let $(\UU_\myF,\VV_\myF,\WW_\myF)$ and $({\UU'}_\myF,{\VV'}_\myF,{\WW'}_\myF)$ be two $F$-PDs of a matrix multiplication tensor $\PhiMat$.
We would like to know whether they are equivalent (see Section~\ref{sec-invar}) and, if they are, to compute the invariance transformations connecting $(\UU_\myF,\VV_\myF,\WW_\myF)$ to $({\UU'}_\myF,{\VV'}_\myF,{\WW'}_\myF)$.
At first, we assume that the permutation transformation is given and we focus on the computation of the scaling and trace transformations.
(We will see in the next subsection how we can compute this permutation transformation without trying all permutations of $\{1,\ldots,F\}$.)
Under mild assumptions on $(\UU_\myF,\VV_\myF,\WW_\myF)$ and $({\UU'}_\myF,{\VV'}_\myF,{\WW'}_\myF)$, we will see how to do this computation using linear algebra only.

First, we make two important comments.
In the sequel, we will always assume that the rank-$1$ terms\footnote{%
$f_r\otimes g_r$ denotes the \emph{tensor product} of functions $f_r$ and $g_r$, i.e., $(f_r\otimes g_r)(u,v)=f_r(u)g_r(v)$.%
}
$(f_r\otimes g_r)w_r$, $1\leq r\leq F$, in the $F$-PDs \eqref{eq-FPD} are \emph{linearly independent}.
Indeed, if one of the rank-$1$ terms $(f_r\otimes g_r)w_r$ can be decomposed as a linear combination of the other rank-$1$ terms, then it is easy to build a polyadic decomposition \eqref{eq-FPD} of $\PhiMat$ with $F-1$ terms (which would be extremely lucky and never happened in the numerical experiments we performed; see Section~\ref{sec-num}).\footnote{Even in this eventuality, if the decompositions $(\UU_\myF,\VV_\myF,\WW_\myF)$ and $({\UU'}_\myF,{\VV'}_\myF,{\WW'}_\myF)$ are equivalent and if the rank-$1$ terms of $(\UU_\myF,\VV_\myF,\WW_\myF)$ are linearly dependent, then the rank-$1$ terms of $({\UU'}_\myF,{\VV'}_\myF,{\WW'}_\myF)$ are also linearly dependent.
Moreover, the $F_*$-term polyadic decompositions $(\UU_{[F^*]},\VV_{[F^*]},\WW_{[F^*]})$ and $({\UU'}_{[F^*]},{\VV'}_{[F^*]},{\WW'}_{[F^*]})$, $1\leq F_*<F$, obtained by removing in $(\UU_\myF,\VV_\myF,\WW_\myF)$ the linearly dependent rank-$1$ terms and the \emph{corresponding terms} in $({\UU'}_\myF,{\VV'}_\myF,{\WW'}_\myF)$, are equivalent and contain both linearly independent rank-$1$ terms.
Hence, modulo a little extra work (due to the non-uniqueness in the choice of linearly dependent terms to remove), we may always reduce to the case where $(\UU_\myF,\VV_\myF,\WW_\myF)$ contains only linearly independent rank-$1$ terms.}

The second comment is summarized in the following theorem:

\begin{theorem}\label{thm-rank-fact-matr}
Let $\PhiMat$ be a matrix multiplication tensor and let $(\UU_\myF,\VV_\myF,\WW_\myF)$ be an $F$-PD of $\PhiMat$.
Let $I \subseteq \{1,\ldots,F\}$ be a subset of indices with $\lvert I\rvert + n \geq F + 1$.
Then the family $\{\UU_r\}_{r\in I}$ fully spans $\Re^{p\times m}$.
\end{theorem}

\begin{proof}
Suppose, on the contrary, that there exists a subset $I \subseteq \{1,\ldots,F\}$ with size $\ell$ such that $\ell+n \geq F+1$ and $\spann(\{\UU_r\}_{r\in I})\neq\Re^{p\times m}$.
Then there exists $\AA_*\in\Re^{m\times p}\setminus\{0\}$ such that $\trace(\UU_r\AA_*)=0$ for each $r\in I$.
Denote by $Y$ the vector space of $p\times n$ matrices $\BB$ such that $\trace(\VV_r\BB)=0$ for each $r\in F\setminus I$.
Since $\lvert F\setminus I\rvert=F-\ell\leq n-1$, we have that $\dime(Y)\geq pn-n+1$.

Now define $Z$ as the vector space of all matrices $\ZZ=(\CC\AA_*)^\top\in\Re^{p\times n}$ for some $\CC\in\Re^{n\times m}$.
Since $\AA_*\neq0$, the dimension of $Z$ is at least $n$.
Now let $\BB\in Y$ and $\ZZ=(\CC\AA_*)^\top\in Z$ and observe that
\begin{align*}
\trace(\BB\ZZ^\top) &= \trace(\AA_*\BB\CC) = \trace\big(\PhiMat(\AA_*,\BB)\CC\big) \\
&= \sum_{r=1}^F \trace(\UU_r\AA_*) \, \trace(\VV_r\BB) \, \trace(\WW_r\CC) .
\end{align*}
From the definitions of $\AA_*$ and $Y$, we conclude that $\trace(\BB\ZZ^\top)=0$.
Thus $Y\cap Z=\{0\}$, so $\dime(Y)+\dime(Z) = pn-n+1+n>pn$.
This contradicts $Y,Z\subseteq\Re^{p\times n}$.
\end{proof}

\begin{remark}\label{rem-rank-fact}
Note that Theorem~\ref{thm-rank-fact-matr} applies, \emph{mutatis mutandis}, to $\VV_\myF$ and $\WW_\myF$.
To see this, it suffices to observe that if $(\UU_\myF,\VV_\myF,\WW_\myF)$ is an $F$-PD of $\PhiMat$ then $(\WW_\myF,\UU_\myF,\VV_\myF)$ and $(\VV_\myF,\WW_\myF,\UU_\myF)$ provide $F$-PDs of $\Phi_{n,m,p}$ and $\Phi_{p,n,m}$ respectively.~\qedfill
\end{remark}

Among other conclusions of this theorem, we get that $\spann(\{\UU_r\}_{1\leq r\leq F})=\Re^{p\times m}$.
Indeed, it suffices to apply Theorem~\ref{thm-rank-fact-matr} with $I = \{1,\ldots,F\}$.
Similar conclusions hold for $\VV_\myF$ and $\WW_\myF$.

We now present the algorithm to compute the scaling and trace transformations between $(\UU_\myF,\VV_\myF,\WW_\myF)$ and $({\UU'}_\myF,{\VV'}_\myF,{\WW'}_\myF)$ or conclude that no such transformations exist.
To simplify the notation, it will be useful to consider $\UU_r$, $\VV_r$ and $\WW_r$ as column vectors and gather them into matrices.
Therefore, we define
\begin{equation}\label{eq-UUt}
\UUt=[\vecc(\UU_1),\ldots,\vecc(\UU_F)] \in \Re^{pm\times F}
\end{equation}
where $\vecc(\cdot)$ is the vectorization (column stacking) operator.
Similarly, we define $\VVt\in\Re^{np\times F}$ and $\WWt\in\Re^{mn\times F}$, and also $\UUt' \in \Re^{pm\times F}$, $\VVt'\in\Re^{np\times F}$ and $\WWt'\in\Re^{mn\times F}$.
The algorithm is guaranteed to work if we make the following assumption on $(\UU_\myF,\VV_\myF,\WW_\myF)$:

\begin{assumption}\label{ass-cluster}
Let $\UUt$, $\VVt$ and $\WWt$ be defined as above.
We assume that either $\UUt$, $\VVt$ or $\WWt$ has clustering number equal to one.
\end{assumption}

\begin{remark}
It is not difficult to see that if $(\UU_\myF,\VV_\myF,\WW_\myF)$ and $({\UU'}_\myF,{\VV'}_\myF,{\WW'}_\myF)$ are equivalent, then $\cl(\UUt)=\cl(\UUt')$, $\cl(\VVt)=\cl(\VVt')$ and $\cl(\WWt)=\cl(\WWt')$.
Thus the clustering numbers (which can be efficiently computed) already offer us a way to eliminate $F$-PDs that are not equivalent.
Therefrom, Assumption~\ref{ass-cluster} can be rephrased (without loss of generality) as follows: either $\cl(\UUt)=\cl(\UUt')=1$, or $\cl(\VVt)=\cl(\VVt')=1$, or $\cl(\WWt)=\cl(\WWt')=1$.~\qedfill
\end{remark}

We will see in Section~\ref{sec-num-equi} that Assumption~\ref{ass-cluster} is satisfied for $100\%$ of the randomly computed samples on which we have performed numerical experiments.
The goal of the algorithm presented in this subsection is to compute matrices $\PP \in \Gl(m)$, $\QQ \in \Gl(p)$ and $\RR \in \Gl(n)$, and scaling coefficients $\lambda_1,\ldots,\lambda_F$, $\mu_1,\ldots,\mu_F$ and $\nu_1,\ldots,\nu_F$ such that $\lambda_r\mu_r\nu_r=1$ and
\begin{equation}\label{eq-scale-trace}
\lambda_r^{}\UU'_r=\QQ^{-1}\UU_r^{}\PP ,\; \mu_r^{}\VV'_r=\RR^{-1}\VV_r^{}\QQ ,\; \nu_r^{}\WW'_r=\PP^{-1}\WW_r^{}\RR
\end{equation}
for every $1\leq r\leq F$.
The above conditions are nonlinear in $\PP$, $\QQ$, $\RR$, $\lambda_r$, $\mu_r$ and $\nu_r$.
However, relying on the assumptions on $(\UU_\myF,\VV_\myF,\WW_\myF)$ and $({\UU'}_\myF,{\VV'}_\myF,{\WW'}_\myF)$, these conditions can be reduced to linear matrix equations.

First of all, we show that the requirement $\lambda_r\mu_r\nu_r=1$ can be dropped.
Indeed, suppose that $(\UU_\myF,\VV_\myF,\WW_\myF)$ and $({\UU'}_\myF,{\VV'}_\myF,{\WW'}_\myF)$ satisfy \eqref{eq-scale-trace}, and let $f'_r$, $g'_r$ and $w'_r$ be given by \eqref{eq-trace} with $({\UU'}_\myF,{\VV'}_\myF,{\WW'}_\myF)$.
Also let $\Phi''$ be given by \eqref{eq-FPD} and \eqref{eq-trace} with $\UU''_r=\QQ^{-1}\UU_r^{}\PP$, $\VV''_r=\RR^{-1}\VV_r^{}\QQ$, $\WW''_r=\PP^{-1}\WW_r^{}\RR$.
Then $\Phi''=\PhiMat$ (trace transformations) and
\[
\Phi''(u,v) = \sum_{r=1}^F \lambda_r^{}\mu_r^{}\nu_r^{} \, f'_r(u) g'_r(v) w'_r .
\]
From the linear independence assumption on the rank-$1$ terms $(f'_r\otimes g'_r)w'_r$, $1\leq r\leq F$, we conclude that $\lambda_r\mu_r\nu_r=1$ is trivially satisfied if \eqref{eq-scale-trace} holds.

According to Assumption~\ref{ass-cluster}, we assume for the rest of this subsection that $\cl(\UUt)=\cl(\UUt')=1$.
We denote by $\AA\otimes\BB$ the \emph{Kronecker product} of two matrices $\AA$ and $\BB$, and we will use the following property of the vectorization operator:
\[
\vecc(\AA\XX\BB) = (\BB^\top\!\otimes\AA)\,\vecc(\XX) .
\]
Then the first equation of \eqref{eq-scale-trace} is equivalent to
\begin{equation}\label{eq-equi-U}
(\PP^\top\!\otimes\QQ^{-1})\UUt = \UUt'\,\diag(\lambda_1,\ldots,\lambda_F) .
\end{equation}
Considering $\PP^\top\!\otimes\QQ^{-1}$ as a single matrix $\MM\in\Re^{pm\times pm}$, \eqref{eq-equi-U} becomes
\begin{equation}\label{eq-equi-U-lin}
\MM\UUt = \UUt'\,\diag(\lambda_1,\ldots,\lambda_F)
\end{equation}
which is linear in $\MM$ and $\lambda_{[F]}=(\lambda_1,\ldots,\lambda_F)$.
The fact that no unwanted solutions are created by this linearization is shown in the following developments.

Let $\AA$ and $\AA'$ be two $m\times n$ matrices with full row-rank, containing no zero columns and with $\cl(\AA)=\cl(\AA')=1$.
Then consider the linear system
\begin{equation}\label{eq-lin-equi}
\MM \AA = \AA' \, \diag(\xi_1,\ldots,\xi_n)
\end{equation}
with variables $\MM\in\Re^{m\times m}$ and $\xi_{[n]} = (\xi_1,\ldots,\xi_n)\in\Re^n$.
This problem is close to problem \eqref{eq-lin-cluster} except that we allow $\AA\neq\AA'$.
Let $\SSc$ be the vector space of $(\MM,\xi_{[n]})$ that are solutions of \eqref{eq-lin-equi}.

\begin{lemma}\label{lem-dim-equi}
Let $\SSc$ be defined as above.
If $\SSc$ contains a solution $(\MM,\xi_{[n]})$ such that $\xi_i \neq 0$ for every $1\leq i\leq n$, then $\dime(\SSc)=1$.
\end{lemma}

\begin{proof}
Let $(\MM,\xi_{[n]})$ be a solution of \eqref{eq-lin-equi} with $\xi_i \neq 0$ for every $1\leq i\leq n$.
We have assumed that $\AA'$ has full row-rank and thus $\AA' \, \diag(\xi_1,\ldots,\xi_n)$ has full row-rank as well.
Hence, $\MM$ must be invertible.
In a similar way as in the proof of Lemma~\ref{lem-lin-cluster-full}, we may assume without loss of generality that the first $m$ columns of $\AA$ span $\Re^m$.
Hence, the first $m$ columns of $\AA'$ span $\Re^m$ too.
We conclude the proof with a similar reasoning as for the first part of the proof of Lemma~\ref{lem-lin-cluster-full}.
\end{proof}

Hence, two cases can happen when solving \eqref{eq-equi-U}: (i)~either the linearized system \eqref{eq-equi-U-lin} admits no solutions with $\lambda_r\neq0$ for every $1\leq r\leq F$; in this case, we conclude that the two $F$-PDs are not (scaling+trace)-equivalent; or (ii)~the solution space $\SSc$ of \eqref{eq-equi-U-lin} is one-dimensional and thus taking an arbitrary nonzero $(\MM,\lambda_\myF)\in\SSc$, it is easy to check whether $\MM$ has the form $\MM=\PP^\top\!\otimes\QQ^{-1}$ for some $\PP\in\Gl(m)$ and $\QQ\in\Gl(p)$.
If the latter does not hold, then the two $F$-PDs are not (scaling+trace)-equivalent.
Otherwise, $\PP$ and $\QQ$ are the \emph{unique} (up to a scalar multiplication) matrices involved in the invariance transformations \eqref{eq-scale-trace}.

Now that we have determined $\PP$ and $\QQ$, we consider the following linear system:
\begin{equation}\label{eq-lin-V-W}
\left\lbrace\begin{array}{r@{\;\,}c@{\;\,}ll}
\RR\VV'_r &=& \tilde{\mu}_r^{}\VV_r^{}\QQ & \quad \text{for all $1\leq r\leq F$} , \\[2pt]
\WW_r'\RR &=& \tilde{\nu}_r^{}\PP\WW_r^{} & \quad \text{for all $1\leq r\leq F$} ,
\end{array}\right.
\end{equation}
where the unknowns are $\RR\in\Re^{n\times n}$ and $\tilde{\mu}_{[F]},\tilde{\nu}_{[F]}\in\Re^F$ for $1\leq r\leq F$.
If \eqref{eq-lin-V-W} admits no solutions $(\RR,\tilde{\mu}_\myF,\tilde{\nu}_\myF)$ with $\tilde{\mu}_r\neq0$ and $\tilde{\nu}_r\neq0$ for every $1\leq r\leq F$, then we conclude that $(\UU_\myF,\VV_\myF,\WW_\myF)$ and $({\UU'}_\myF,{\VV'}_\myF,{\WW'}_\myF)$ are not (scaling+trace)-equivalent.
On the other hand, if $\tilde{\mu}_r\neq0$ and $\tilde{\nu}_r\neq0$ for every $1\leq r\leq F$, then $\RR$ is invertible because $\spann(\{\VV_r\QQ\}_{1\leq r\leq F})=\Re^{n\times p}$ (see Remark~\ref{rem-rank-fact}).
The $6$-tuple $(\PP,\QQ,\RR,\lambda_\myF,\mu_\myF,\nu_\myF)$ with $\mu_r^{}=\tilde{\mu}_r^{-1}$ and $\nu_r^{}=\tilde{\nu}_r^{-1}$ then provides a solution to the (scaling+trace)-equivalence problem \eqref{eq-scale-trace}.

%%%%%%%%%%%%%%%%%%%%%%%%%%%%%%%%%%%%%%%%%%%%%%%%%%%%%%%%%%%%%%%%%%%%%%%%%%%%%%%%%%%%%%%%%%%%%%%%%%%%
\subsection{Computation of the permutation transformation}\label{subsec-perm}

In the previous subsection, we have described a procedure to compute the scaling and trace transformations connecting two $F$-PDs $(\UU_\myF,\VV_\myF,\WW_\myF)$ and $({\UU'}_\myF,{\VV'}_\myF,{\WW'}_\myF)$ or conclude that no such transformations exist.
The equivalence of $(\UU_\myF,\VV_\myF,\WW_\myF)$ and $({\UU'}_\myF,{\VV'}_\myF,{\WW'}_\myF)$ can then be decided in finite time by trying every permutation $\sigma\in\Perm_F$ and testing the (scaling+trace)-equivalence of $\sigma((\UU_\myF,\VV_\myF,\WW_\myF))$\footnote{where $\sigma((\UU_\myF,\VV_\myF,\WW_\myF)) = (\sigma(\UU_\myF),\sigma(\VV_\myF),\sigma(\WW_\myF))$ and $\sigma(\XX_\myF)$ is the permuted $F$-uple $(\XX_{\sigma(1)},\ldots,\XX_{\sigma(F)})$.} and $({\UU'}_\myF,{\VV'}_\myF,{\WW'}_\myF)$.
Due to the combinatorial growth of $\lvert\Perm_F\rvert$, an exhaustive exploration of $\Perm_F$ is generally not feasible in practice.
In this section, we explain how to efficiently decide whether the two $F$-PDs are equivalent without trying all permutations $\sigma\in\Perm_F$.

\begin{definition}\label{def-simul-sim}
Let $\AA_{[m]}=(\AA_1,\ldots,\AA_m)$ and $\BB_{[m]}=(\BB_1,\ldots,\BB_m)$ be two ordered sets of $n\times n$ matrices.
We say that $\AA_{[m]}$ and $\BB_{[m]}$ are \emph{simultaneously similar} if there exists $\XX\in\Gl(n)$ such that $\AA_i=\XX^{-1}\BB_i\XX$ for every $1\leq i\leq m$.
\end{definition}

Let $(\UU_\myF,\VV_\myF,\WW_\myF)$ and $({\UU'}_\myF,{\VV'}_\myF,{\WW'}_\myF)$ be two $F$-PDs of the matrix multiplication tensor $\PhiMat$.
For each $1\leq r\leq F$, define the matrices $\MM_r^{}=\WW_r^{}\VV_r^{}\UU_r^{}$ and $\MM'_r=\WW'_r\VV'_r\UU'_r$.
If $\sigma((\UU_\myF,\VV_\myF,\WW_\myF))$ and $({\UU'}_\myF,{\VV'}_\myF,{\WW'}_\myF)$ are (scaling+trace)-equivalent for some $\sigma\in\Perm_F$, then from \eqref{eq-scale-trace} we have
\begin{equation}\label{eq-simul-simil}
\begin{split}
\MM'_r &= \lambda_r^{}\mu_r^{}\nu_r^{}\,\MM'_r \\
&= (\PP^{-1}\WW_{\sigma(r)}\RR) (\RR^{-1}\VV_{\sigma(r)}\QQ) (\QQ^{-1}\UU_{\sigma(r)}\PP) = \PP^{-1}\MM_{\sigma(r)}\PP.
\end{split}
\end{equation}
In other words, $\sigma(\MM_\myF)$ and ${\MM'}_\myF$ are simultaneously similar.

We define a \emph{partial permutation} of $\{1,\ldots,F\}$ as any injective function $\pi$ from $I\subseteq\{1,\ldots,F\}$ into $\{1,\ldots,F\}$.
We say that $\pi$ \emph{coincides} with the (total) permutation $\sigma\in\Perm_F$ if $\pi(r)=\sigma(r)$ for every $r\in I$.
If $\sigma$ is as in \eqref{eq-simul-simil} and $\pi$ coincides with $\sigma$, then it is clear that
\begin{equation}\label{eq-sim}
\pi\big((\MM_r^{})_{r\in I}^{}\big) \;\;\;\text{and}\;\;\; (\MM'_r)_{r\in I}^{} \;\;\;\text{are simultaneously similar.}
\end{equation}

The following notation will be useful for the description of the algorithm for computing $\sigma$.
For $F'\in\{0,\ldots,F\}$, we denote by $\Inj(F',F)$ the set of injective functions from $\{1,\ldots,F'\}$ into $\{1,\ldots,F\}$.
Each function of $\Inj(F',F)$ is seen as a subset of $\{1,\ldots,F'\}\times\{1,\ldots,F\}$.
The length of $\pi\in\Inj(F',F)$ is simply $\lvert\pi\rvert=F'$, and the range of $\pi\in\Inj(F',F)$ is defined as $\Range(\pi)=\{\,\pi(r) : 1\leq r\leq F'\,\}$.

The idea behind the algorithm to compute $\sigma$ is the following.
First, we start from a partial permutation $\pi\in\Inj(F',F)$ with $F'$ small.
We check whether $\pi$ is susceptible to coincide with $\sigma$ by checking whether \eqref{eq-sim} is satisfied or not (see also Remark~\ref{rem-sim}).
If \eqref{eq-sim} is satisfied, then we try to extend $\pi$ to a larger partial permutation $\pi^+=\pi\cup \{(F'+1,\ell)\}$ with $\ell\in\{1,\ldots,F\}\setminus\Range(\pi)$.
We check again whether $\pi^+$ is susceptible to coincide with $\sigma$ according to \eqref{eq-sim}.
If this is the case, we repeat the process with $\pi^+$.
Otherwise, we try other extensions $\pi\cup \{(F'+1,\ell')\}$.
If all possible extensions $\pi\cup \{(F'+1,\ell)\}$, $\ell\in\{1,\ldots,F\}\setminus\Range(\pi)$, have been tried and none of them coincides with $\sigma$, then we restart the process with the restriction $\pi^-=\pi\vert_{\{1,\ldots,F'-1\}}\in\Inj(F'-1,F)$ and try to extend $\pi^-$ to $\pi^-\cup\{(F',\ell)\}$ with $\ell\in\{1,\ldots,F\}\setminus\Range(\pi)$.

When we reach a full permutation $\pi\in\Inj(F,F)=\Perm_F$, then we can decide whether the permuted decomposition $\pi(\UU_\myF,\VV_\myF,\WW_\myF)$ and the decomposition $({\UU'}_\myF,{\VV'}_\myF,{\WW'}_\myF)$ are (scaling+trace)-equivalent using the procedure of the previous subsection.
If they are, then we have found the correct permutation transformation between $(\UU_\myF,\VV_\myF,\WW_\myF)$ and $({\UU'}_\myF,{\VV'}_\myF,{\WW'}_\myF)$.
Otherwise, we continue to search for another permutation $\pi$.

When the algorithm terminates, if the two $F$-PDs are equivalent, the algorithm is guaranteed to give the corresponding scaling, trace and permutation transformations.
If they are not equivalent, the algorithm will also detect it because all permutations $\pi\in\Perm_F$ will be rejected: either because the partial permutation $\pi\vert_{\{1,\ldots,F'-1\}}$ has been rejected previously in the algorithm, or because $\pi$ does not lead to (scaling+trace)-equivalent decompositions.
Clearly, the computational savings (compared to trying all permutations) are interesting if most of the ``incorrect'' permutations $\pi$ are rejected in a early stage, i.e., $\pi\vert_{\{1,\ldots,F'-1\}}$ is rejected for $F'\ll F$.
The computational aspects are discussed in the paragraphs below.

\newcommand{\forcond}{$i=0$ \KwTo $n$}
\SetKwProg{Fn}{Function}{}{end function}
\SetKwFunction{FRecurs}{FnRecursive}
\SetKwData{EquivFound}{EquivFound}
\SetAlgoLongEnd
\newcommand{\mytrue}{\mbox{\upshape true}}
\newcommand{\myfalse}{\mbox{\upshape false}}

We have implemented the algorithm as the recursive function described in Algorithm~\ref{algo-equiv}.
The recursive function must be called with $(\pi,b)=\FRecurs{$\varnothing,\myfalse$}$.
If the output $b$ is true, then the two decompositions are equivalent and the permutation transformation is given by $\pi$.
On the other hand, if $b$ is false, then the two $F$-PDs are not equivalent.

\begin{algorithm2e}[h]
\vspace{2pt}
\KwData{$\pi\in\bigcup_{n=0}^F \Inj(n,F)$ and $b$ is a boolean.}
\vspace{4pt}
\Fn{\FRecurs{$\pi,b$}}{
\uIf{$b=\mytrue$}{
    Return $(\pi,b)$\;
}
\uElseIf(\tcc*[f]{$[\star]$}){$\lvert\pi\rvert=F$}{
    \eIf(\tcc*[f]{$[\clubsuit]$}){$\pi(\UU_\myF,\VV_\myF,\WW_\myF)$ and $({\UU'}_\myF,{\VV'}_\myF,{\WW'}_\myF)$ are (scaling+trace)-equivalent}{
        Return $(\pi,\mytrue)$\;
    }{
        Return $(\varnothing,\myfalse)$\;
    }
}
\Else{
    \ForEach{$\ell\in\{1,\ldots,F\}\setminus\Range(\pi)$}{
        Let $\pi^+=\pi\cup \{(\lvert\pi\rvert+1,\ell)\}$\;
        \If{\eqref{eq-sim} holds with $\pi^+$}{
            Let $(\pi',b')=\FRecurs{$\pi^+,b$}$\;
            \If{$b'=\mytrue$}{
                Return $(\pi',\mytrue)$\;
            }
        }
    }
    Return $(\varnothing,\myfalse)$\tcc*[r]{$[\spadesuit]$}
}
}
\caption{Recursive function to decide whether two $F$-PDs are equivalent.}
\label{algo-equiv}
\end{algorithm2e}

\begin{remark}\label{rem-sim}
Checking the simultaneous similarity of $\AA_{[m]}=(\AA_1,\ldots,\AA_m)$ and $\BB_{[m]}=(\BB_1,\ldots,\BB_m)$ can be approached by solving a linear system
\[
\XX\AA_i - \BB_i\XX = 0 \qquad\text{for all}\quad 1\leq i\leq m,
\]
with unknown $\XX\in\Re^{n\times n}$, and check whether there exists a solution $\XX$ that is invertible.
However, this approach is not efficient and not robust to rounding errors.
Therefore, we have used a different approach.
Consider scalar coefficients $\alpha_1,\ldots,\alpha_m\in\Re$.
A necessary condition for $\AA_{[m]}$ and $\BB_{[m]}$ to be simultaneously similar is that $\sum_{i=1}^m \alpha_i\AA_i$ and $\sum_{i=1}^m \alpha_i\BB_i$ have the same eigenvalues counted with multiplicity.
By doing this for randomly generated sets of coefficients $\alpha_1,\ldots,\alpha_m\in\Re$, this gives a very efficient way to check the simultaneous similarity of $\AA_{[m]}$ and $\BB_{[m]}$ with high probability.~\qedfill
\end{remark}

\begin{remark}\label{rem-algo-general}
Strictly speaking, the use of Algorithm~\ref{algo-equiv} supposes that Assumption~\ref{ass-cluster} is satisfied.
One could wonder whether we can still obtain some information from Algorithm~\ref{algo-equiv} even if the assumption is not satisfied.
The answer is yes.
We modify the algorithm as follows.
If condition $[\star]$ is satisfied, then instead of testing whether the $F$-PDs are (scaling+trace)-equivalent, we directly output $(\pi,\mytrue)$ and exit the function.
With this modified algorithm, if the call of the function $(\pi,b)=\FRecurs{$\varnothing,\myfalse$}$ returns the value $b=\mytrue$, then we cannot say anything about the equivalence of the two $F$-PDs.
However, if $b=\myfalse$, then we are sure that the two $F$-PDs are not equivalent.~\qedfill
\end{remark}

Numerical experiments for the algorithm described in this section are presented in Section~\ref{sec-num-equi}.
Regarding the complexity of the algorithm, the computation of the scaling and trace transformations relies only on solving linear systems of equations.
The system \eqref{eq-equi-U-lin} consists of $mpF$ equations with $(mp)^2+F$ variables.
Because $F\geq mp$ (consequence of Theorem~\ref{thm-rank-fact-matr}), the complexity of solving \eqref{eq-equi-U-lin} is at most $\BigO([Fmp]^3)$.
Similarly, solving \eqref{eq-lin-V-W} requires at most $\BigO([F(pn+nm)]^3)$.
Therefore, the complexity of the (scaling+trace)-equivalence part of the algorithm is bounded by $\BigO([F\max\{mp,pn,nm\}]^3)$.

The complexity of the permutation computation part is more difficult to evaluate.
It is obviously bounded by $F!$.
Hence, an upper bound for the global complexity of the algorithm is $\BigO(F![F\max\{mp,pn,nm\}]^3)$.
However, in all the numerical experiments we have performed (see Section~\ref{sec-num-equi}), it appears that Algorithm~\ref{algo-equiv} never reaches step $[\star]$ more than once.
In fact, all the partial permutations $\pi$ for which the algorithm reaches step $[\spadesuit]$ satisfy $\lvert \pi\rvert \leq 9$ (see Table~\ref{tab-equiv}--Depth).
In other words, whenever $\pi\in\Inj(F',F)$ could not lead to a correct permutation, then the algorithm detected it rapidly.
Hence, in practice (for our numerical experiments), the computational complexity of the complete algorithm is $\BigO([F\max\{mp,pn,nm\}]^3)$.

%%%%%%%%%%%%%%%%%%%%%%%%%%%%%%%%%%%%%%%%%%%%%%%%%%%%%%%%%%%%%%%%%%%%%%%%%%%%%%%%%%%%%%%%%%%%%%%%%%%%
\section{Characteristic polynomials and discretizable decompositions}\label{sec-discrete}

Drawing upon the simultaneous similarity property \eqref{eq-simul-simil} of equivalent decompositions, we introduce a simple necessary criterion for a decomposition of a matrix multiplication tensor to be equivalent to a discrete decomposition.

\begin{definition}\label{def-discretizable}
A decomposition $(\UU_\myF,\VV_\myF,\WW_\myF)$ is \emph{discretizable} if it is equivalent to a discrete decomposition $({\UU'}_\myF,{\VV'}_\myF,{\WW'}_\myF)$.
[Clearly, it is necessary and sufficient to require that $(\UU_\myF,\VV_\myF,\WW_\myF)$ is only (scaling+trace)-equivalent to $({\UU'}_\myF,{\VV'}_\myF,{\WW'}_\myF)$.]
\end{definition}

We refer the reader to Section~\ref{sec-preli} for the definition and relevance of discrete decompositions in the context of fast matrix multiplication.
Numerical algorithms for computing polyadic decompositions of matrix multiplication tensors do not lead in general to solutions of this kind.
The possibility to transform a general decomposition into a discrete one using invariance transformation opens the door to a new generation of algorithms to compute discrete solutions relying on a two-step approach (first compute a general decomposition and then discretize it).
However, it is not clear when a decomposition can be discretized with invariance transformations so that the two-step approach may be inapplicable in some cases.
The aim of this section is \emph{not} to describe algorithms for transforming general decompositions into discrete decompositions but we propose a necessary criterion for a decomposition to be discretizable.

The criterion draws upon the observations made in Section~\ref{subsec-perm}: if $(\UU_\myF,\VV_\myF,\WW_\myF)$ and $({\UU'}_\myF,{\VV'}_\myF,{\WW'}_\myF)$ are (scaling+trace)-equivalent, then the families $\MM_\myF$ and ${\MM'}_\myF$, defined by $\MM_r^{}=\WW_r^{}\VV_r^{}\UU_r^{}$ and $\MM'_r=\WW'_r\VV'_r\UU'_r$, are simultaneously similar (Definition~\ref{def-simul-sim}).
In particular, $\sum_{r=1}^F \beta_r\MM_r$ and $\sum_{r=1}^F \beta_r\MM'_r$ are also similar for every coefficients $\beta_r\in\Re$ (cf.~Remark~\ref{rem-sim}) and thus they have the same characteristic polynomial.

Assume that $({\UU'}_\myF,{\VV'}_\myF,{\WW'}_\myF)$ is a discrete $F$-PD.
Then $\UU'_r\in(q\ZZb)^{p\times m}$, $\VV'_r\in(q\ZZb)^{n\times p}$ and $\WW'_r\in(q\ZZb)^{m\times n}$ for some $q\in\Re$.
Hence, $\MM'_r\in(q^3\ZZb)^{m\times m}$ for every $1\leq r\leq F$.
Let the coefficients $\beta_r$ in the paragraph above be integers.
If we denote the characteristic polynomial of $\frac{1}{q^3}\sum_{r=1}^F \beta_r\MM_r$ by
\begin{equation}\label{eq-charpol}
\begin{split}
p(t) &= p(t;\beta_1,\ldots,\beta_F) \\
&= \det\left(tI-\frac{1}{q^3}\sum_{r=1}^F \beta_r\MM_r\right) = t^m + \alpha_{m-1} t^{m-1} + \ldots + \alpha_0,
\end{split}
\end{equation}
it is not hard to see that the coefficients $\alpha_i\in\ZZb$ for every $0\leq i<m$.

\begin{definition}\label{def-char}
Let the matrices $\MM_r$ be defined as above.
We say that the decomposition $(\UU_\myF,\VV_\myF,\WW_\myF)$ satisfies the \emph{discretizability criterion with parameter $q$} if for every integer coefficients $\beta_r$, $1\leq r\leq F$, the coefficients of the characteristic polynomial \eqref{eq-charpol} satisfy $\alpha_i\in \ZZb$ for every $0\leq i< m$.
\end{definition}

From the developments above, it is clear that satisfying the discretizability criterion with some parameter $q\in\Re$ is a necessary condition for being discretizable.
In the following section, we will see that most of the sample decompositions on which we have performed numerical experiments do not satisfy the discretizability criterion with $q=1$ or $q=1/2$ for tensors larger than the $2\times2$ by $2\times2$ case, contrasting with the abundance in the literature of discrete decompositions with $q=1$ or $q=1/2$ for these tensors (see also Section~\ref{sec-intro}).

%%%%%%%%%%%%%%%%%%%%%%%%%%%%%%%%%%%%%%%%%%%%%%%%%%%%%%%%%%%%%%%%%%%%%%%%%%%%%%%%%%%%%%%%%%%%%%%%%%%%
\section{Numerical experiments}\label{sec-num}

We have applied the results of Sections~\ref{sec-equiv} and~\ref{sec-discrete} on large sample sets of decompositions for matrix multiplication tensors up to the $m=p=n=3$ case.
The goal is to get for the first time a view on the distributions of essentially unique decompositions and the distributions of discretizable decompositions: how many essentially unique decompositions do there exist?
If two different decompositions are computed with a numerical algorithm, are they likely to be equivalent?
Likely to be discretizable for some given $q$?

The way to obtain these samples is described in the next subsection.
The reason we restrict to cases smaller than or equal to the $m=p=n=3$ case is explained in the next subsection as well.
All computations were performed in Matlab.
The computation-intensive part, namely, the generation of the samples, was executed on a Linux machine with 28 cores and 128 GBytes of RAM.
The other computations were done on a laptop having 4 cores and 16 GBytes of RAM running Linux.

%%%%%%%%%%%%%%%%%%%%%%%%%%%%%%%%%%%%%%%%%%%%%%%%%%%%%%%%%%%%%%%%%%%%%%%%%%%%%%%%%%%%%%%%%%%%%%%%%%%%
\subsection{Computing polyadic decompositions}

In the numerical experiments, we considered the six different cases $(m,p,n;F)$ summarized in Table~\ref{tab-cases-num} (first two columns), where $(m,p,n)$ is the size of the matrix multiplication tensor and $F$ is the number of rank-$1$ terms, i.e., we considered $F$-PDs of $\PhiMat$.
For each $(m,p,n)$, the associated $F$ is the smallest $F$ for which we know there exists in the literature a decomposition of $\PhiMat$ with $F$ terms (see, e.g., \cite{BalBen16,Smi13}).%
\footnote{%
Note that for the first four cases in Table~\ref{tab-cases-num}, $F$ is equal to the rank of the associated tensor and thus cannot be decreased; see, e.g., \cite[Chapter~15]{BurCla13} for the $(1,2,1)$ and $(2,1,2)$ cases; for $(2,2,2)$, see \cite[Theorem~11]{BroDob78}; and for $(2,3,2)$, see \cite{Ale15}.
For the $(3,2,3)$ case, the best known lower bound on the rank of $\Phi_{3,2,3}$ is $\rank(\Phi_{3,2,3})\geq14$ (see, e.g., \cite[Theorem~11]{BroDob78}), and for $(3,3,3)$ the best known lower bound is $\rank(\Phi_{3,3,3})\geq19$, shown by Bl\"aser \cite{Bla03}.
However, no $F$-term polyadic decompositions of $\Phi_{3,2,3}$ and $\Phi_{3,3,3}$ with $F<15$ and $F<23$ respectively are known for the moment.%
}
For each case, we want to obtain large sets of decompositions on which to apply the results of Sections~\ref{sec-equiv} and~\ref{sec-discrete}.

Computing polyadic decompositions of matrix multiplication tensors is notoriously difficult (see, e.g., \cite{Smi13,TicPha17} and references therein).
Quite a few papers in the literature about tensor decompositions are devoted to this specific problem.
For the numerical experiments of this paper, we have used the method proposed by Tichavsk{\'y} et al.~\cite{TicPha17} to compute $N_s=10\,000$ samples (decompositions) for the six cases listed in Table~\ref{tab-cases-num}.
For an alternative method, we refer the reader to \cite{Smi13}.
See also \cite{BenBal15}.
We have used $\UU_{r,0}\in\Re^{p\times m}$, $\VV_{r,0}\in\Re^{n\times p}$ and $\WW_{r,0}\in\Re^{m\times n}$ with entries chosen uniformly at random in $[-1,1]$ as initial iterates for Tichavsk{\'y} et al.'s method.
The method does not always converge to a global minimum; hence we sometimes had to try more than one initial iterate to converge to an exact solution (the third and fourth columns give an idea of the effort required to compute the $N_s$ decompositions).
In the end, we have at our disposal for each case $N_s$ samples of $F$-term polyadic decompositions of $\PhiMat$.
We denote them by $({\UU^\kappa}_\myF,{\VV^\kappa}_\myF,{\WW^\kappa}_\myF)$ with $\kappa\in\{1,\ldots,N_s\}$.

In the numerical computations, the tensors $\Phi\in\Bil(\Re^{m\times p},\Re^{p\times n};\Re^{m\times n})$ are represented by the three-dimensional arrays $\PhiT\in\Re^{mp\times pn\times nm}$ obtained from the canonical identifications $\Re^{m\times p}\cong\Re^{mp}$, etc.
Regarding floating-point arithmetic limitations, a sample $({\UU^\kappa}_\myF,{\VV^\kappa}_\myF,{\WW^\kappa}_\myF)$ is considered as an $F$-PD of $\PhiMat$ if
\[
\lvert\PhiT^\kappa(i,j,k)-\PhiT_{m,p,n}(i,j,k)\rvert< 10^{-9} \qquad \forall i,j,k
\]
where $\Phi^\kappa$ is the tensor defined by \eqref{eq-FPD} and \eqref{eq-trace} with $({\UU^\kappa}_\myF,{\VV^\kappa}_\myF,{\WW^\kappa}_\myF)$.

\begin{table}
\centering
\renewcommand{\arraystretch}{1.1}
\begin{tabular}{@{}ll@{\hspace{35pt}}l@{\hspace{25pt}}l@{}}
\toprule
$(m,p,n)$ & $F$ & \# trials & Elapsed time [hours] \\
\midrule
$(1,2,1)$ & $2$ & $10\,000$ & $0.017$ \\
$(2,1,2)$ & $4$ & $10\,000$ & $0.05$ \\
$(2,2,2)$ & $7$ & $10\,762$ & $0.32$ \\
$(2,3,2)$ & $11$ & $10\,133$ & $0.51$ \\
$(3,2,3)$ & $15$ & $18\,003$ & $40$ \\
$(3,3,3)$ & $23$ & $15\,829$ & $16.8$ \\
\bottomrule
\end{tabular}
\vspace{3pt}
\caption{First and second columns: different cases considered in the numerical experiments.
Fourth column: total time required to compute the $N_s=10\,000$ decompositions with Tichavsk{\'y} et al.'s method \cite{TicPha17}.
Third column: number of randomly generated initial guesses (trials) $(\UU_{\myF,0},\VV_{\myF,0},\WW_{\myF,0})$ we had to use to compute the $N_s=10\,000$ samples.}
\label{tab-cases-num}
\end{table}

\begin{remark}\label{rem-limit333}
For matrix multiplication tensors larger than the $(3,3,3)$ case, it becomes very difficult to compute polyadic decompositions of these tensors: the global convergence of the algorithm decreases significantly while the cost for a single iteration of Tichavsk{\'y} et al.'s method grows as $\BigO([F(mp+pn+nm)]^3)$.
It becomes thus unrealistic to compute large sets of decompositions for these tensors.~\qedfill
\end{remark}

%%%%%%%%%%%%%%%%%%%%%%%%%%%%%%%%%%%%%%%%%%%%%%%%%%%%%%%%%%%%%%%%%%%%%%%%%%%%%%%%%%%%%%%%%%%%%%%%%%%%
\subsection{Discretizable decompositions}\label{sec-num-disc}

We start with the analysis of the discretizability property of the sample decompositions $({\UU^\kappa}_\myF,{\VV^\kappa}_\myF,{\WW^\kappa}_\myF)$.
For instance, we would like to find the decompositions that are not equivalent to a discrete decomposition with $q=1/2$ (Definition~\ref{def-discr}).
To do this, we will apply the necessary criterion for discretizability with parameter $1/2$.

For every $\kappa\in\{1,\ldots,N_s\}$, let $p(t)=p(t;\beta_1,\ldots,\beta_F)$ be as in \eqref{eq-charpol} where the $\beta_r$'s are randomly chosen integer coefficients.
In our experiments, we use $16$ sets of coefficients sampled uniformly at random in $\{-5,-4,\ldots,5\}^F$, providing thus $16$ polynomials
\[
p^\kappa_j(t) = t^m + \alpha^\kappa_{j,m-1} t^{m-1} + \ldots + \alpha^\kappa_{j,0},\qquad j=1,\ldots,16.
\]
For each $\kappa\in\{1,\ldots,N_s\}$, we let
\[
\ND_\kappa = \max_{1\leq j\leq 16} \:\max_{0\leq i< m} \:\lvert \alpha_{j,i}^\kappa - \round(\alpha_{j,i}^\kappa) \rvert
\]
where $\round(\alpha)$ is the closest integer to $\alpha$.
The value of $\ND_\kappa$ is thus a measure of how close are the polynomials $p^\kappa_j(t)$ to polynomials with integer coefficients.
From the results of Section~\ref{sec-discrete}, a decomposition $({\UU^\kappa}_\myF,{\VV^\kappa}_\myF,{\WW^\kappa}_\myF)$ for which $\ND_\kappa$ is (significantly) nonzero is not equivalent to a discrete decomposition with $q=1/2$.

The histograms in Figure~\ref{fig-int-dev} show the distribution of decompositions based on the value of $\ND_\kappa$.
More precisely, each bar of the histograms represents the number of decompositions $({\UU^\kappa}_\myF,{\VV^\kappa}_\myF,{\WW^\kappa}_\myF)$ with $\ND_\kappa$ in the corresponding range.
For the $(1,2,1)$, and $(2,2,2)$ cases, we observe that, for all the decompositions, the polynomials $p^\kappa_j(t)$ have \emph{integer} coefficients (within a very small tolerance).
Hence, $100\%$ of the decompositions satisfy the discretizability criterion with $q=1/2$.
This is not surprising since all the $2$-PDs (resp.~$7$-PDs) of $\Phi_{1,2,1}$ (resp.~$\Phi_{2,2,2}$) are equivalent (see \cite{deG78b} and Remark~\ref{rem-equiv-theory}), and $\Phi_{1,2,1}$ (resp.~$\Phi_{2,2,2}$) admits a discrete decomposition with $q=1$.\footnote{%
For $\Phi_{1,2,1}$, take, e.g.,
\vspace{-7pt}
\begin{align*}
\UU_1 = [1,0],\quad \VV_1 = [1,0]^\top,\quad \WW_1 = 1, \\
\UU_2 = [0,1],\quad \VV_2 = [0,1]^\top,\quad \WW_2 = 1,
\end{align*}
(see also Remark~\ref{rem-equiv-theory}).
For $\Phi_{2,2,2}$, take, e.g., Strassen's algorithm.%
}

In contrast, for the $(2,1,2)$, $(2,3,2)$, $(3,2,3)$ and $(3,3,3)$ cases, we observe that most of the decompositions do not satisfy the necessary criterion for discretizability with $q=1/2$.
This implies that most of the decompositions are not equivalent to a discrete decompositions with parameter $q=1/2$.
This last observation has to be put in contrast with the abundance of decompositions $(\UU_\myF,\VV_\myF,\WW_\myF)$ for which the entries of $\UU_r$, $\VV_r$ and $\WW_r$ belong to $\{0,\pm1/2,\pm1\}$ in the literature \cite{BalBen16,Lad76,TicPha17,OhJin13}.

\begin{figure}
\centering
\begin{tabular}{@{}c@{\;\;\;}c@{}}
$(1,2,1)$ & $(2,1,2)$ \\[2pt]
\includegraphics[height=4.5cm]{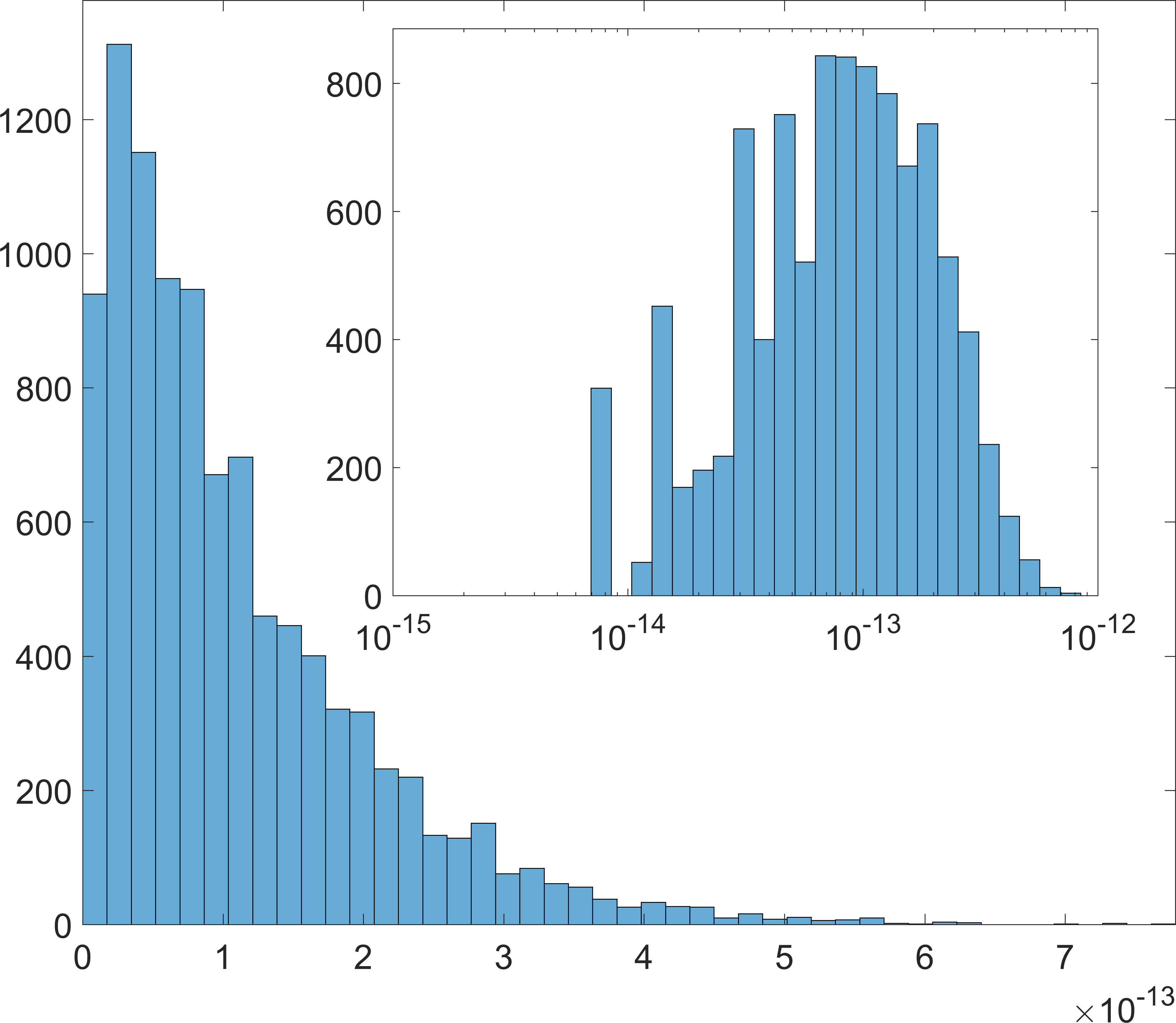} &
\includegraphics[height=4.5cm]{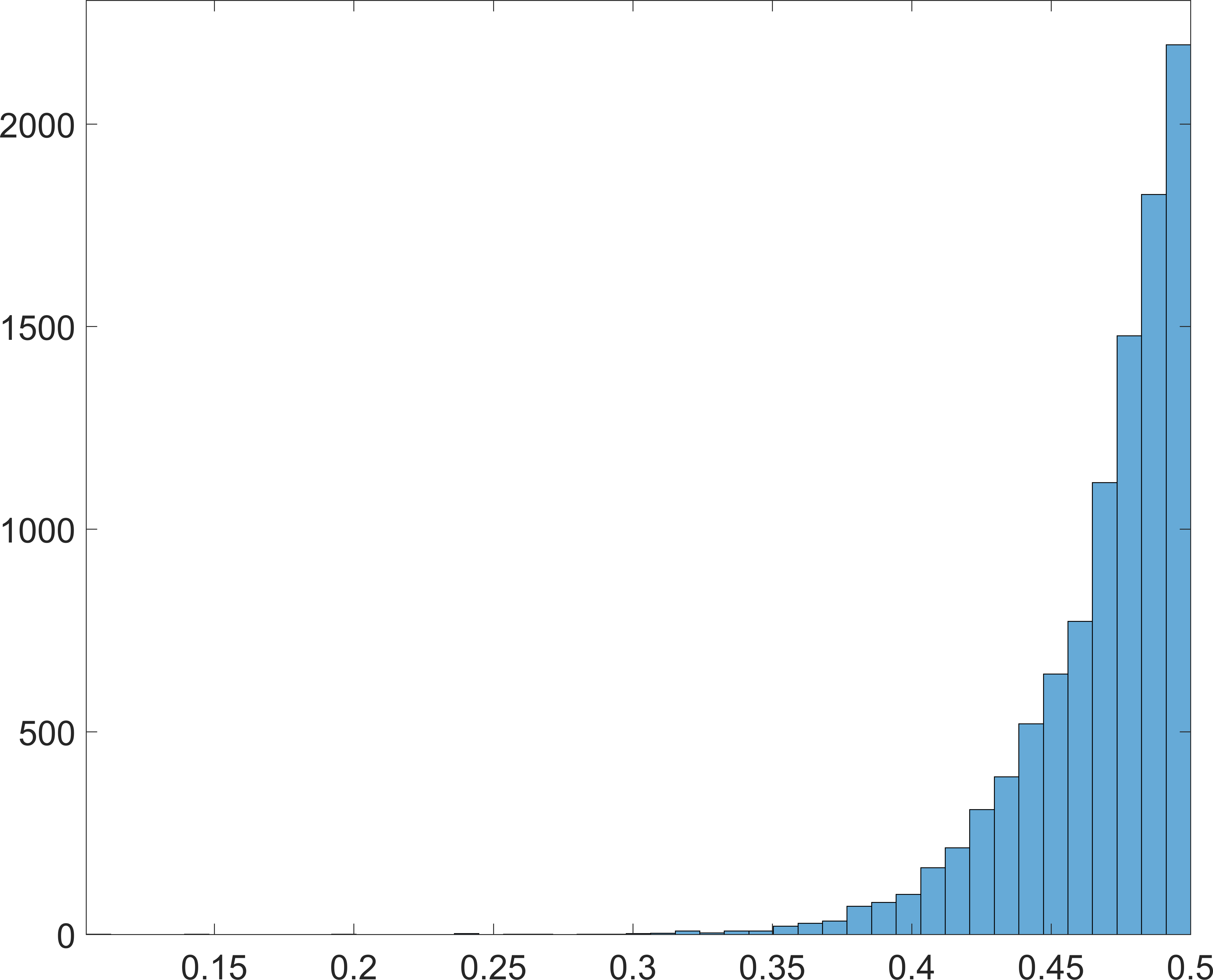} \\[5pt]
$(2,2,2)$ & $(2,3,2)$ \\[2pt]
\includegraphics[height=4.5cm]{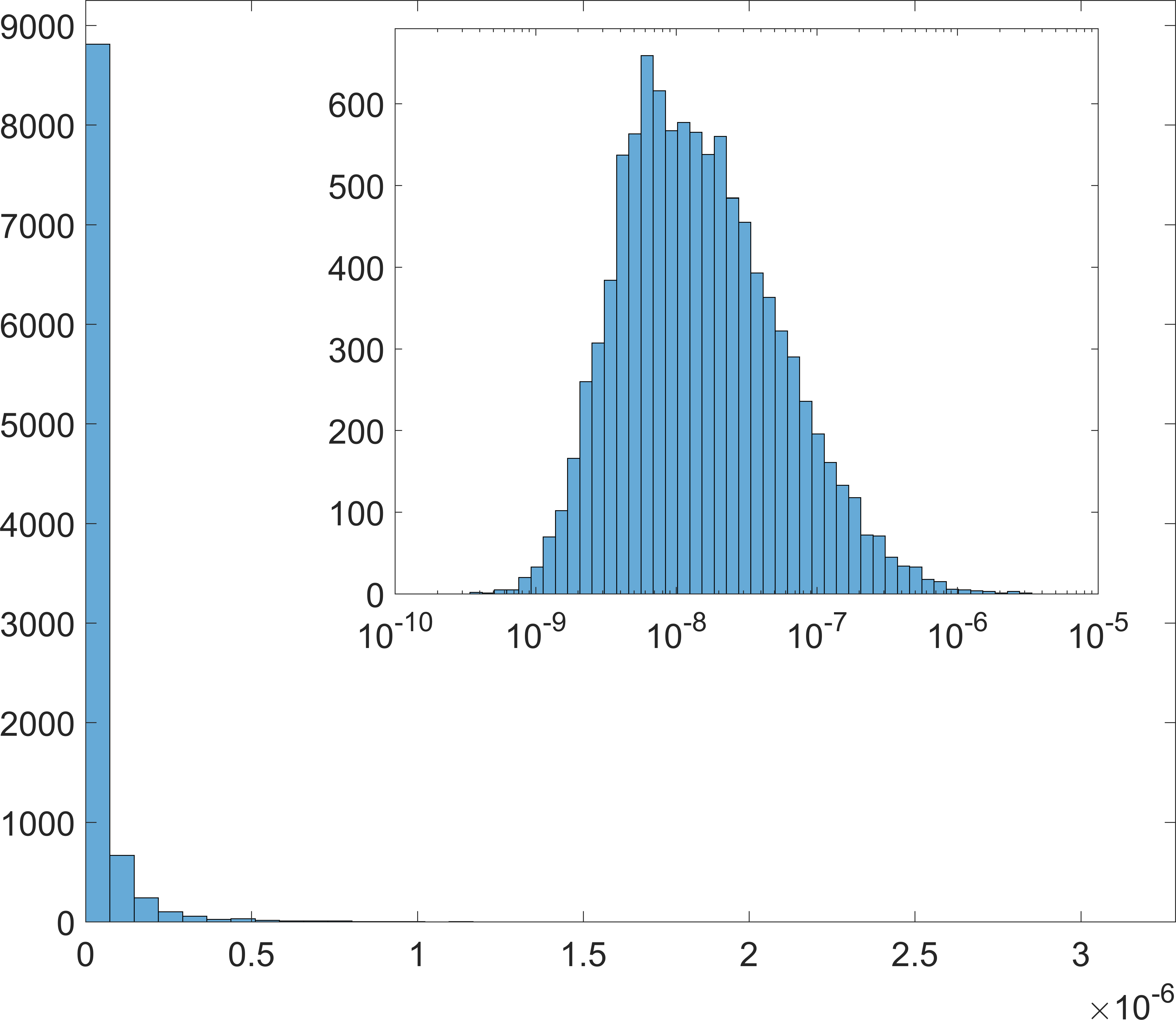} &
\includegraphics[height=4.5cm]{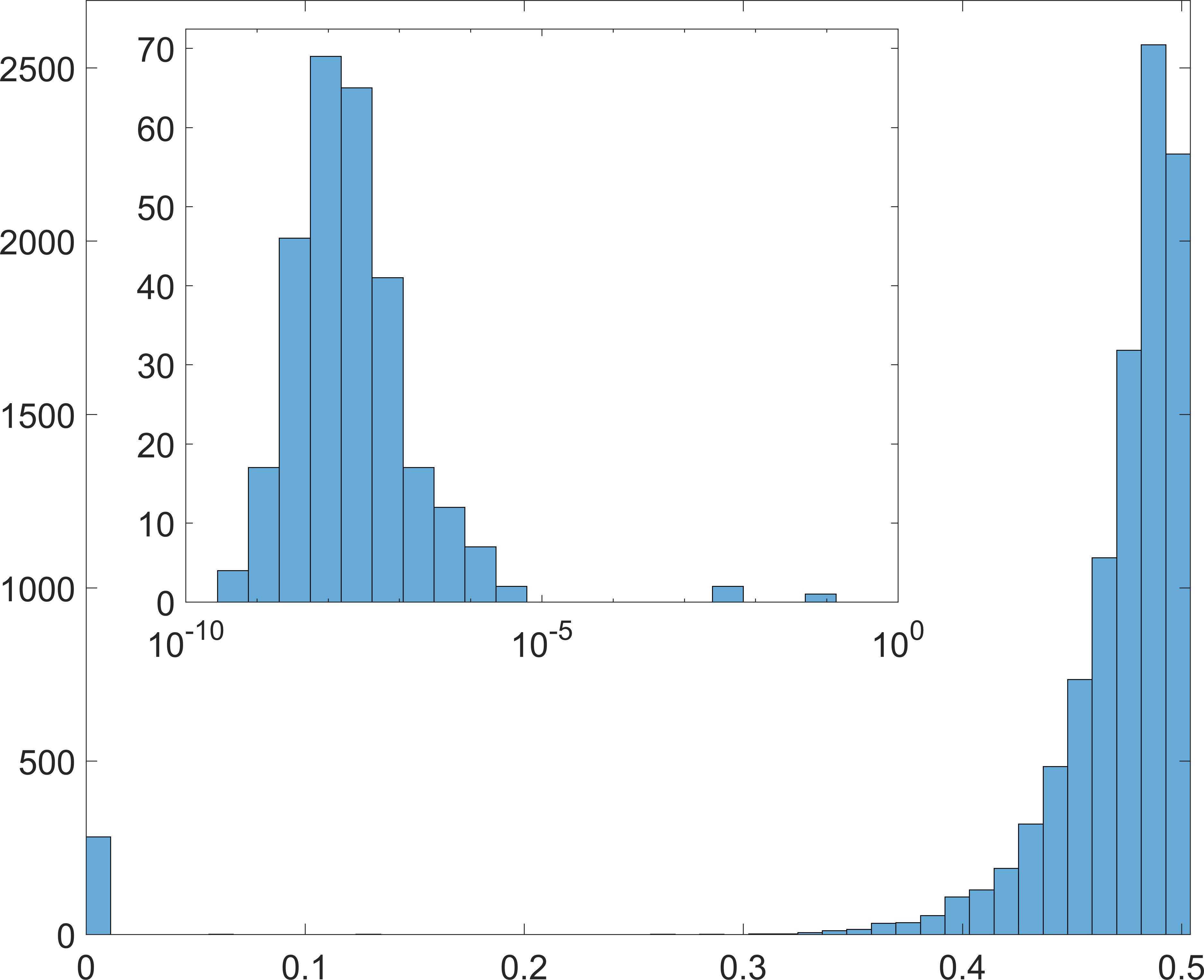} \\[5pt]
$(3,2,3)$ & $(3,3,3)$ \\[2pt]
\includegraphics[height=4.5cm]{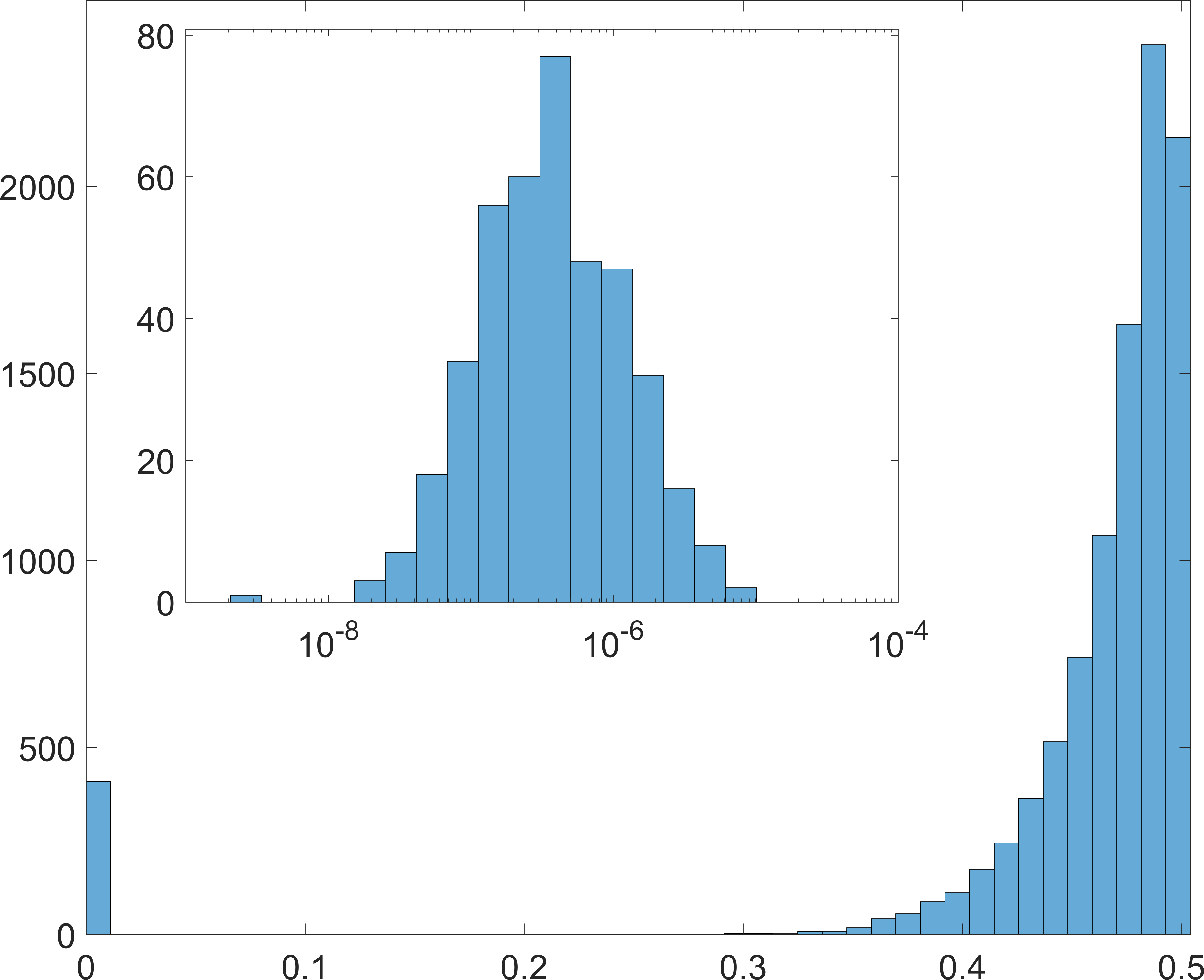} &
\includegraphics[height=4.5cm]{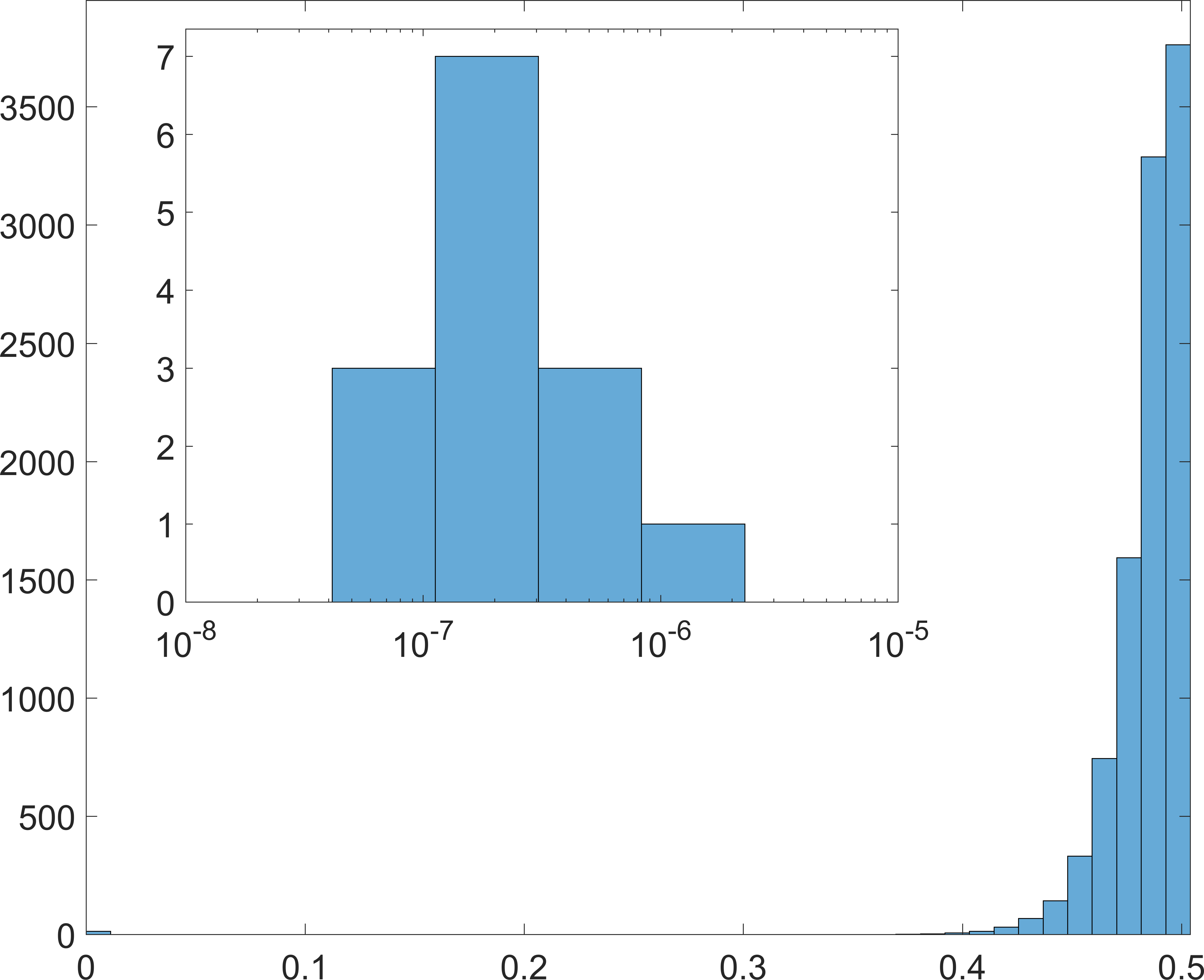}
\end{tabular}
\caption{Distribution of decompositions based on how close the polynomials $p^\kappa_j(t)$ are to characteristic polynomials with integer coefficients.
Horizontal axis: $\ND_\kappa$.
Vertical axis: \# decompositions $({\UU^\kappa}_\myF,{\VV^\kappa}_\myF,{\WW^\kappa}_\myF)$ with $\ND_\kappa$ in the corresponding range.
Remember that the total number of decompositions is equal to $N_s=10\,000$.
The insets provide a zoom on the decompositions $({\UU^\kappa}_\myF,{\VV^\kappa}_\myF,{\WW^\kappa}_\myF)$ with $\ND_\kappa<0.1$.
Note the logarithmic scale of the horizontal axis in the inset.}
\label{fig-int-dev}
\end{figure}

Further experiments can be conducted to investigate the discretizability of the decompositions with respect to other parameters $q$.
However, due to space limitations, we do not present them in this paper.

%%%%%%%%%%%%%%%%%%%%%%%%%%%%%%%%%%%%%%%%%%%%%%%%%%%%%%%%%%%%%%%%%%%%%%%%%%%%%%%%%%%%%%%%%%%%%%%%%%%%
\subsection{Equivalence classes of decompositions}\label{sec-num-equi}

In the previous subsection, we have seen that, except for the $(1,2,1)$ and $(2,2,2)$ cases, most of the decompositions are not equivalent to a discrete decomposition with coefficients in $\{0,\pm1/2,\pm1\}$.
In this section, we will analyze the pairwise equivalence of the decompositions.
This will reveal the distributions of the equivalence classes among the sample sets of decompositions.
Therefore, we use the algorithm developed in Section~\ref{sec-equiv}.

First, in order to apply Algorithm~\ref{algo-equiv}, we need to ensure that Assumption~\ref{ass-cluster} is satisfied for every decomposition $({\UU^\kappa}_\myF,{\VV^\kappa}_\myF,{\WW^\kappa}_\myF)$, $\kappa\in\{1,\ldots,N_s\}$.
Therefore, for each decomposition $({\UU^\kappa}_\myF,{\VV^\kappa}_\myF,{\WW^\kappa}_\myF)$, we have computed (using Theorem~\ref{thm-lin-cluster-gen}) the \emph{clustering vector} of the decomposition defined as the vector $[\cl(\UUt^\kappa),\cl(\VVt^\kappa),\cl(\WWt^\kappa)]$.
The results are summarized in Figure~\ref{fig-cluster-number}.
As we can see, for each case, $100\%$ of the decompositions have at least one matrix $\UUt^\kappa$, $\VVt^\kappa$ or $\WWt^\kappa$ with clustering number equal to one and thus satisfy Assumption~\ref{ass-cluster}.

\begin{figure}
\centering
\begin{tabular}{@{}c@{\;}c@{\;}c@{}}
$(1,2,1)$ & $(2,1,2)$ & $(2,2,2)$ \\[-8pt]
\tikz[baseline=(current bounding box.north)] \node[anchor=north]{\includegraphics[width=4cm]{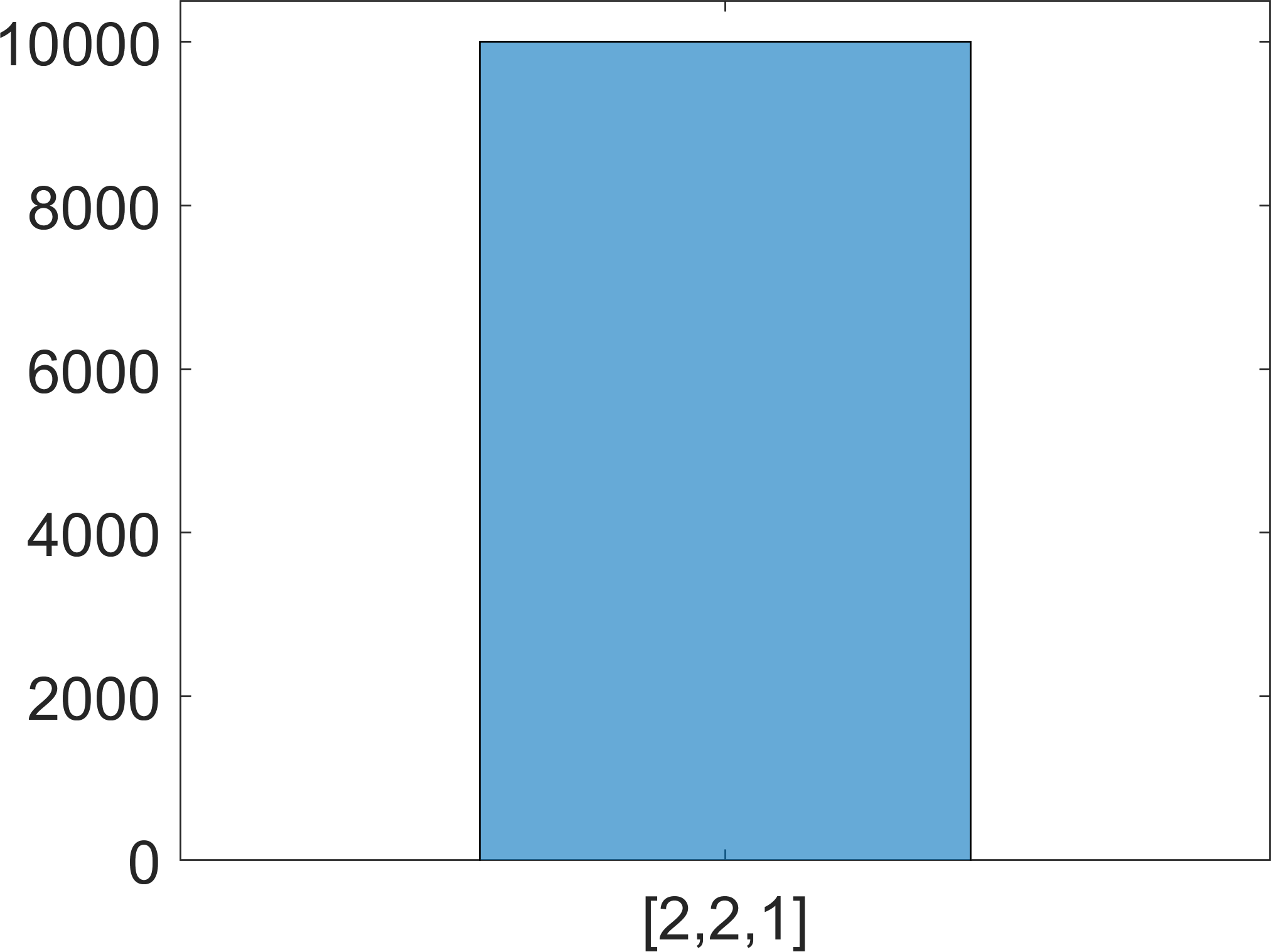}}; &
\tikz[baseline=(current bounding box.north)] \node[anchor=north]{\includegraphics[width=4cm]{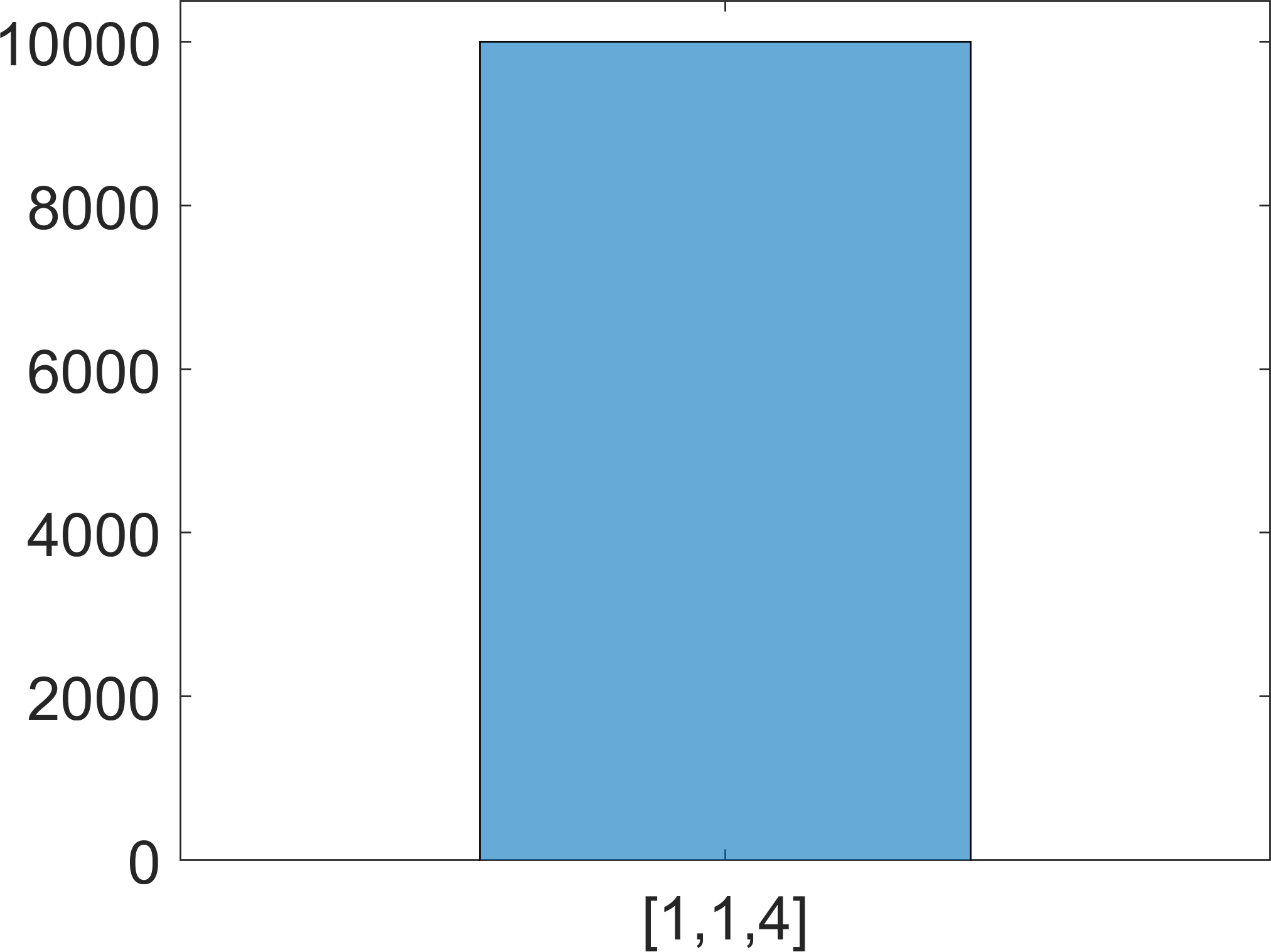}}; &
\tikz[baseline=(current bounding box.north)] \node[anchor=north]{\includegraphics[width=4cm]{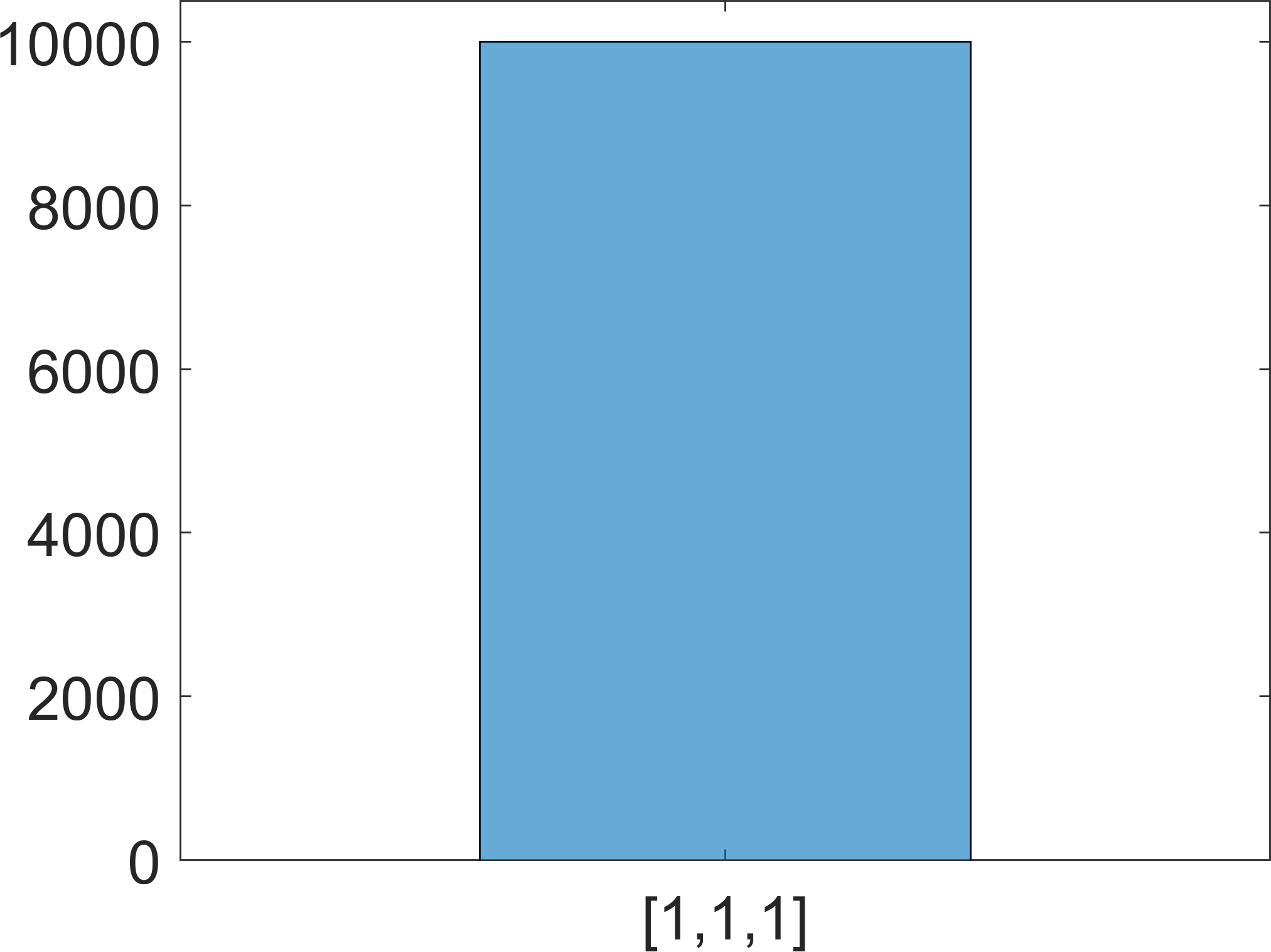}}; \\
\\[-6pt]
$(2,3,2)$ & $(3,2,3)$ & $(3,3,3)$ \\[-8pt]
\tikz[baseline=(current bounding box.north)] \node[anchor=north]{\includegraphics[width=4cm]{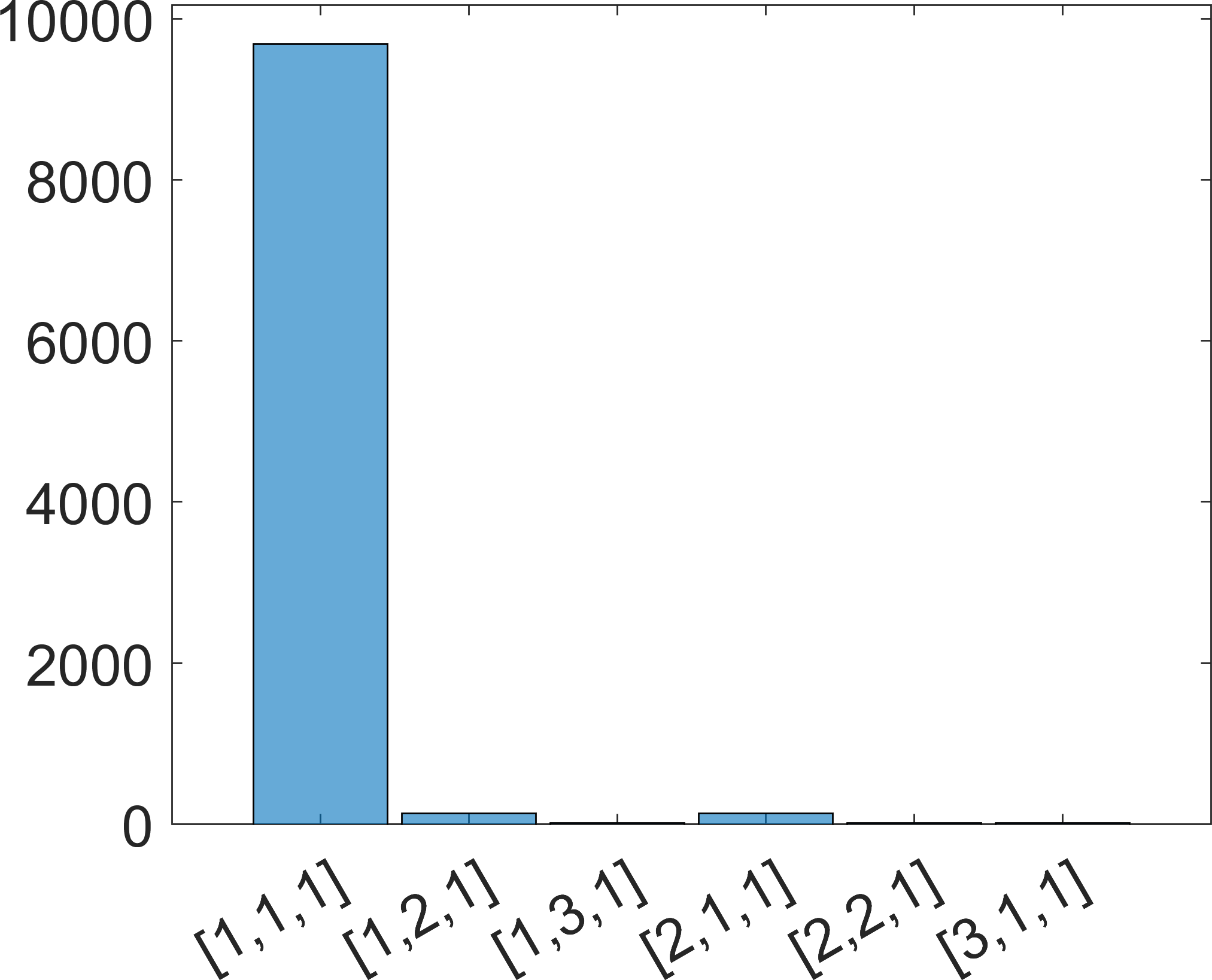}}; &
\tikz[baseline=(current bounding box.north)] \node[anchor=north]{\includegraphics[width=4cm]{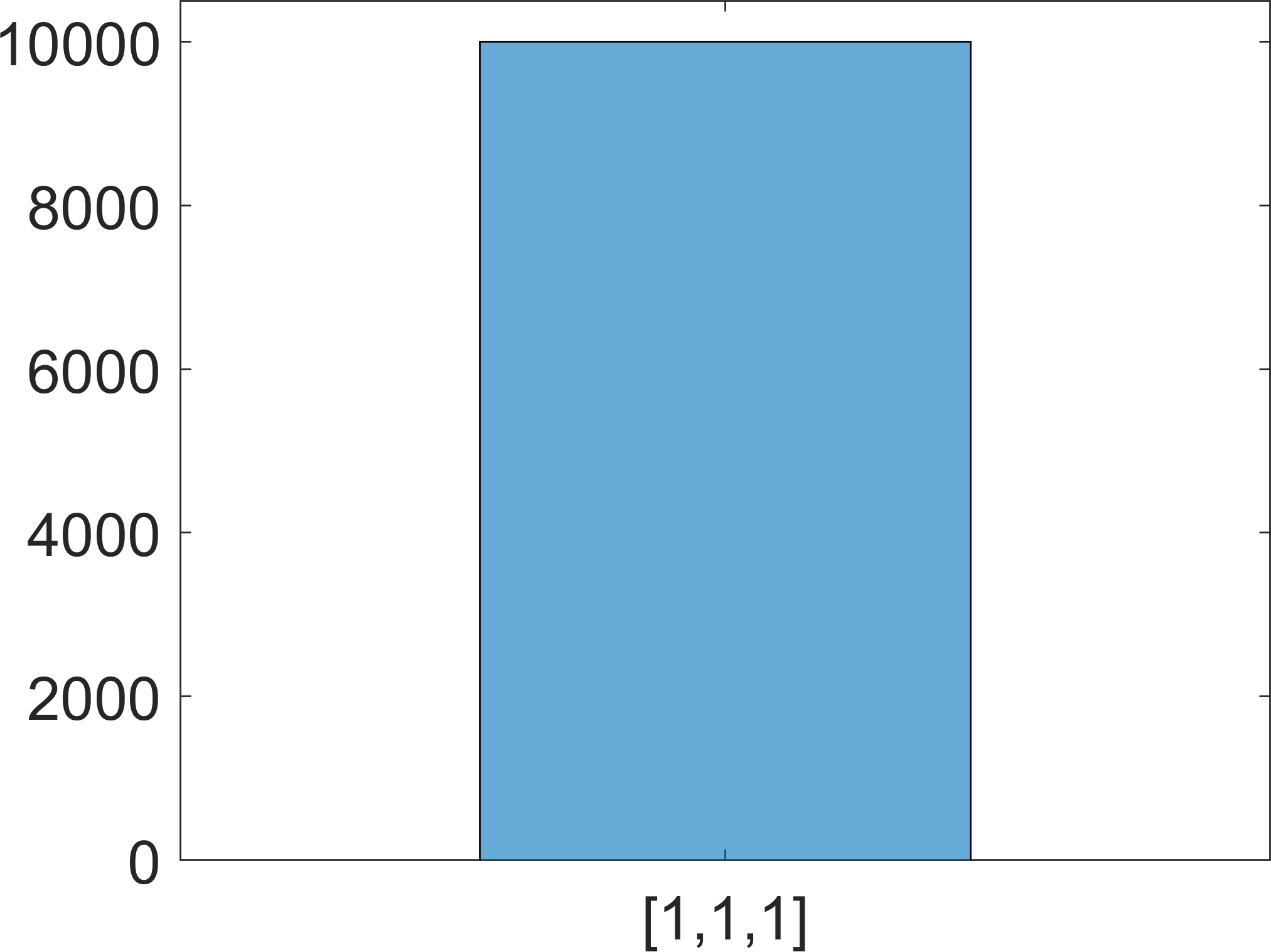}}; &
\tikz[baseline=(current bounding box.north)] \node[anchor=north]{\includegraphics[width=4cm]{cluster_number_vectors_3x2x3}};
\end{tabular}
\caption{Clustering vectors $[\cl(\UUt^\kappa),\cl(\VVt^\kappa),\cl(\WWt^\kappa)]$ of the decompositions.}
\label{fig-cluster-number}
\end{figure}

We can thus apply Algorithm~\ref{algo-equiv} to check the equivalence between pairs of decompositions for the different cases.
The results are gathered in Table~\ref{tab-equiv}.
We observe that for the $(1,2,1)$ and $(2,2,2)$ cases, every decompositions are pairwise equivalent (see also Remark~\ref{rem-equiv-theory}).
As a consequence, it is not surprising that all the decompositions are discretizable.
This situation is more surprising for $(2,1,2)$, $(2,3,2)$, $(3,2,3)$ and $(3,3,3)$ cases.
For these cases, the decompositions seem to be pairwise equivalent with probability zero.

\begin{table}
\centering
\renewcommand{\arraystretch}{1.2}
\newcommand{\mysep}{3mm}
\begin{tabular}{@{}>{\itshape}l@{\hspace{23pt}}l@{\hspace{15pt}}l@{\hspace{15pt}}@{}ll@{}}
\toprule
&
\multirow{2}{*}{\renewcommand{\arraystretch}{1}\begin{tabular}{@{}l@{}} Percentage of\\equivalent pairs\end{tabular}}
&
\multirow{2}{*}{\renewcommand{\arraystretch}{1}\begin{tabular}{@{}l@{}} Mean elapsed\\time [sec] \end{tabular}}
&
\multicolumn{2}{c}{Depth$^*$}\\
\cmidrule{4-5}
&&& Max. & Mean \\
\midrule
$(1,2,1)$ & $100\%$ & $4.74\cdot10^{-4}$ & 1 & 1 \\
$(2,1,2)$ & $0\%$ & $1.02\cdot10^{-3}$ & 1 & 1 \\
$(2,2,2)$ & $100\%$ & $2.68\cdot10^{-3}$ & 4 & 0.55 \\
$(2,3,2)$ & $0\%$ & $4.94\cdot10^{-3}$ & 6 & 1.08 \\
$(3,2,3)$ & $0\%$ & $2.66\cdot10^{-2}$ & 9 & 2.91 \\
$(3,3,3)$ & $0\%$ & $2.82\cdot10^{-2}$ & 5 & 1.51 \\
\bottomrule
\end{tabular}
\vspace{3pt}
\caption{Equivalence of decompositions.
For each case, we have used Algorithm~\ref{algo-equiv} to check the equivalence between $10\,000$ randomly chosen pairs of decompositions inside the cluster.
The second column gives the percentage of pairs of equivalent decompositions.
The third column gives the average time required to check the equivalence of the decompositions with Algorithm~\ref{algo-equiv}.
$^*$The depth of the algorithm is the maximal length of a partial permutation $\pi\in\Inj(n,F)$ that is rejected (see Algorithm~\ref{algo-equiv}).}
\label{tab-equiv}
\end{table}

The third column of Table~\ref{tab-equiv} gives the average computation time to check the equivalence between two decompositions.
We observe that the algorithm takes no more than $30~\text{ms}$.
In comparison, for the $(3,3,3)$ case for example, the naive method (testing all possible permutations) would have required to test condition $[\clubsuit]$ in Algorithm~\ref{algo-equiv}, which has a complexity of $\BigO([F\max\{mp,pn,nm\}]^3)$, $23!=2.59\cdot10^{22}$ times.

\begin{remark}\label{rem-equiv-theory}
It can be shown that all the decompositions of $\Phi_{1,2,1}$ and $\Phi_{2,2,2}$ (respectively) are pairwise equivalent, which corroborates the results of the numerical experiments using Algorithm~\ref{algo-equiv}.
For the $(2,2,2)$ case, we refer the reader to \cite{deG78b}.
For the $(1,2,1)$ case, let $(\UU_\myF,\VV_\myF,\WW_\myF)$ and $({\UU'}_\myF,{\VV'}_\myF,{\WW'}_\myF)$ be two $2$-PDs of $\Phi_{1,2,1}$.
Observe that $\Phi_{1,2,1}$ maps $2$-dimensional vectors to their scalar product and thus can be represented with the identity matrix:
\[
\Phi_{1,2,1}(u,v) = u^\top\left[
\begin{array}{cc}
1 & 0 \\
0 & 1 \\
\end{array}
\right]v.
\]
Since $(\UU_\myF,\VV_\myF,\WW_\myF)$ is a decomposition, it is not hard to see that
\[
\left[
\begin{array}{cc}
1 & 0 \\
0 & 1 \\
\end{array}
\right] = \UUt \: \diag(\WW_1,\WW_2) \: \VVt^\top
\]
where we remind that $\UUt$ and $\VVt$ are defined as \eqref{eq-UUt}.
Using a scaling transformation, we may assume that $\WW_1=\WW_2=1$.
Hence $\UUt$ and $\VVt^\top$ are inverses of each other, and so are $\UUt'$ and $\VVt'^\top$.
Then let $\PP=\UUt\VVt'^\top$ and observe that
\[
\begin{array}{@{}r@{$\:$}c@{$\:$}c@{$\:$}c@{$\:$}l}
\PP^{-1} \UUt &=& \big[\UUt'\UUt^{-1}\big] \, \UUt &=& \UUt', \\[5pt]
\PP^\top \VVt &=& \big[\VVt'\VVt^{-1}\big] \, \VVt &=& \VVt'. \\
\end{array}
\]
Hence, we have found a trace transformation with $\QQ=\RR=1\in\Gl(1)$ between the two decompositions.~\qedfill
\end{remark}

%%%%%%%%%%%%%%%%%%%%%%%%%%%%%%%%%%%%%%%%%%%%%%%%%%%%%%%%%%%%%%%%%%%%%%%%%%%%%%%%%%%%%%%%%%%%%%%%%%%%
\section{Conclusions}\label{sec-concl}

In this paper, we have described an algorithm for efficiently deciding whether two decompositions of a given matrix multiplication tensor are equivalent through invariance transformations.
We have introduced the notion of clustering number of a matrix and we have demonstrated the correctness of the algorithm provided some conditions on the clustering number of the factor matrices of the decompositions are satisfied.
This condition was satisfied for $100\%$ of the numerical samples on which we have applied our algorithm.

The analysis of the equivalence classes of decompositions is relevant in the context of fast matrix multiplication as it sheds light on the diversity of essentially unique fast matrix multiplication algorithms.
In the numerical experiments we have performed, it appears that two decompositions are equivalent with probability zero (except for the multiplication of $1\times2$ by $2\times2$ matrices and the multiplication of $2\times2$ matrices for which we can prove the essential uniqueness of their decompositions) indicating that there are many essentially different algorithms for the fast multiplication of $3\times3$ matrices for example.

Drawing upon the observation that decompositions with coefficients in a discrete set provide fast matrix multiplication with better performance, we have also provided a necessary criterion for a decomposition to be equivalent to a decomposition with these properties.
We have applied the criterion on numerical samples and observed that the majority of the decompositions do not satisfy the criterion for being equivalent to a decomposition with coefficients in, e.g., $\{0,\pm1\}$ or $\{0,\pm1/2,\pm1\}$.

%%%%%%%%%%%%%%%%%%%%%%%%%%%%%%%%%%%%%%%%%%%%%%%%%%%%%%%%%%%%%%%%%%%%%%%%%%%%%%%%%%%%%%%%%%%%%%%%%%%%
\section*{Acknowledgments}

The authors would like to thank Samuel Fiorini and Gwe\-na\"el Joret for insightful discussions on the link between the clustering number and matroids.

\bibliography{myrefs}

\end{document}